\documentclass[a4paper]{article}
\usepackage[utf8]{inputenc}
\usepackage[margin=2.5cm]{geometry}
\usepackage{enumerate}
\usepackage{amsmath,amsthm,amsfonts,amssymb,mathtools,mathrsfs}
\usepackage{bm}
\usepackage{bbm}
\usepackage{xcolor}
\usepackage{hyperref}

\usepackage[english]{babel}
\addto\extrasenglish{

}

\newcommand{\R}{\mathbb{R}}

\newcommand{\test}{\ccs(\R^4)}

\newcommand{\di}{\text{d}}
\newcommand{\en}[1]{\omega_{\textbf{#1}}}
\newcommand{\cre}[2]{#1^\dagger_{\textbf{#2}}}
\newcommand{\ani}[2]{#1_{\textbf{#2}}}

\newcommand{\bsh}[2]{\mathscr{B}^#1_{\textbf{#2}}}

\newcommand{\bshs}[2]{\mathscr{B}^{#1^{2}}_{\textbf{#2}}}

\newcommand{\Ban}{\mathcal{B}}
\newcommand{\Hil}{\mathcal{H}}
\newcommand{\A}{\mathcal{A}}
\newcommand{\M}{\mathcal{M}}
\newcommand{\Mt}{\tilde{\mathcal{M}}}
\newcommand{\Mh}{\widehat{\mathcal{M}}}
\newcommand{\N}{\mathcal{N}}

\newcommand{\ocal}{\mathcal{O}}
\newcommand{\polyalg}{\mathscr{P}}

\newcommand{\polyn}{\mathcal{P}}

\newcommand{\dom}[1]{\operatorname{dom}{#1}}

\newcommand{\affil}{\mathrel{\eta}}
\newcommand{\dual}{\mathsf{d}}

\newcommand{\idop}{\boldsymbol{1}}

\newcommand{\spt}{\mathfrak{M}}
\newcommand{\csf}{\mathfrak{C}}

\newcommand{\ccs}{\mathcal{C}_\mathrm{c}^\infty}

\newcommand{\vpsi}{\psi}
\newcommand{\vxi}{\xi}
\newcommand{\vphi}{\varphi}

\newcommand{\rw}[0]{\mathcal{W}}

\newcommand{\normord}[1]{{:}#1{:}}

\newcommand{\testrw}[0]{\mathcal{C}^\infty_\mathrm{c}(\rw)}
\newcommand{\testo}[0]{\mathcal{C}^\infty_\mathrm{c}(\ocal)}
\newcommand{\testf}[0]{\mathcal{C}^\infty_\mathrm{c}(\R^4)}

\newcommand{\domSplit}{\mathcal{D}_\mathsf{split}}
\newcommand{\domSchwartz}{\mathcal{D}_{\mathcal{S},0}}
\newcommand{\domP}{\polyalg(\R^4)\Omega}
\newcommand{\domAffil}[1]{\mathcal{S}_\mathsf{affil}(#1)}

\newtheorem{definition}{Definition}[section]
\newtheorem{lemma}[definition]{Lemma}
\newtheorem{proposition}[definition]{Proposition}
\newtheorem{theorem}[definition]{Theorem}

\theoremstyle{definition}

\numberwithin{equation}{section}

\title{Quantum $L^p$ inequalities in thermal states}
\author{Henning Bostelmann\thanks{Merseburg University of Applied Sciences, Eberhard-Leibnitz-Straße 2,
06217 Merseburg, Germany; e-mail: \href{mailto:henning.bostelmann@hs-merseburg.de}{henning.bostelmann@hs-merseburg.de}} 
\and Daniela Cadamuro\thanks{%
University of Leipzig, Institute for Theoretical Physics, Leipzig, Germany; e-mail: \href{mailto:cadamuro@itp.uni-leipzig.de}{cadamuro@itp.uni-leipzig.de}
} 
\and Leonardo Sangaletti\thanks{%
Dipartimento di Fisica, Università di Genova, Italy;
e-mail: \href{mailto:leonardo.sangaletti@edu.unige.it}{leonardo.sangaletti@edu.unige.it}}
}
\date{November 24, 2025}

\begin{document}

\maketitle

\begin{abstract}
In thermal quantum field theory, the global Liouvillian (the generator of time translations) is passive. How is this reflected in the properties of its local density, a quantum field? We propose that the locally averaged density is bounded below, but not above, with respect to the noncommutative $L^4$ norm. This is analogous to the known quantum energy inequalities in the vacuum situation.  Our examples include thermal equilibrium on Minkowski space, and the Unruh effect on the Rindler wedge, both for the real scalar free field.
\end{abstract}

\section{Introduction}\label{sec:intro}

In the vacuum sector of quantum field theory (QFT), the energy operator (the Hamiltonian $H$) is a positive operator, but the energy \emph{density} $h$ as a local quantum field cannot be positive in all states \cite{EGJ}. Nevertheless, in typical situations $h$ still fulfills lower bounds of the following type: For any nonnegative test function $f$, there is a constant $c_f>0$ such that for all (sufficiently regular) state vectors $\vpsi$,
        \begin{equation}\label{eq:qei0}
            (\vpsi, h(f) \vpsi) \geq  -c_f \| \vpsi\|^2.
        \end{equation}
Such \emph{quantum energy inequalities} (QEIs) have been found in a range of linear models, also on curved spacetime backgrounds (see \cite{Fewster:lecturenotes} for a review), and in some simple models with self-interaction \cite{BostelmannCadamuroFewster:ising, BostelmannCadamuroMandrysch:qeiint}. Here $\| \vpsi\|$ on the right-hand side of \eqref{eq:qei0} is understood as the usual Hilbert space norm. In other situations, one may need to replace $ \| \vpsi\|$ with a different norm or other homogeneous expression in $\vpsi$ (``state-dependent QEIs'' \cite{FewsterOsterbrink2008,BostelmannFewster:2009}); this still indicates some reminiscent positivity of $h(f)$ as long as the analogous \emph{upper} bound is violated (``nontriviality'' of the QEI).

In the present paper, we ask whether such QEIs still hold in \emph{thermal} states of QFT, rather than in the vacuum sector.

Depending how this question is posed, the answer might be almost obvious: Thermal equilibrium states of QFT are expected to be locally normal to the vacuum; i.e., when evaluated on local observables such as $h(f)$, they can be replaced with a density matrix in the vacuum sector. Therefore, the QEI \eqref{eq:qei0} -- or its obvious generalization to mixed states -- would automatically extend to thermal equilibrium states and their local excitations. Alternatively speaking, at least in linear theories, equilibrium states fulfill the Hadamard condition with respect to the vacuum \cite{SVW:microlocal}, and therefore the standard argument for QEIs \cite[Sec.~2.4]{Fewster:lecturenotes} applies. In a sense, we have asked whether the local energy density in thermal states is (up to small errors) larger than in the vacuum -- a property that would certainly be expected heuristically.

However, our question in this paper is more subtle: We want to compare the energy density in a local excitation of the equilibrium state to its value in the equilibrium state itself, and ask whether the difference is positive. Stated differently, we ask whether positivity properties of the \emph{Liouvillian} (replacing the Hamiltonian in the thermal case) transfer to its local density. This might seem counterintuitive at first, since excitations of thermal states can contain both ``particles'' and ``holes'' over the background, and hence one would not expect any positivity of the Liouvillian, let alone of its density.

To show that such positivity can nevertheless be expected, we recall the setting of thermal equilibrium states in QFT (\cite{HHW:equilibrium}; see \autoref{sec:tqft} for more detail). Given an observable algebra $\A$ (here: the quasilocal algebra of a QFT) and a thermal equilibrium state $\omega$ of $\A$ (i.e., a state fulfilling the KMS condition for some inverse temperature $\beta$), one constructs two representations $\pi$ and $\tilde \pi$ of $\A$ (one linear, the other anti-linear) on a common Hilbert  space $\Hil$, such that $\pi(\A)$ and $\tilde\pi(\A)$ commute; $\pi(\A)$ may be interpreted as the thermal system in question (the QFT), whereas $\tilde\pi(\A)$ is interpreted as a thermal bath. The state $\omega$ is represented as a vector $\Omega\in\Hil$, a pure (entangled) state of the two combined systems. Time evolution is given by a group of unitaries $U(t)=e^{itL}$,  with its generator the Liouvillian $L$; we may heuristically think that $L=H-\tilde H$ with $H$ the Hamiltonian for $\pi(\A)$ and $\tilde H$ for $\tilde\pi(\A)$, but the individual pieces $H$ and $\tilde H$ lose their meaning in the limit of infinite spatial volume. The corresponding energy densities, $h$ and $\tilde h$, are still valid objects however; and defining $\ell(f) := h(f)-\tilde h(f)$, the field $\ell$ typically is indeed a density for $L$, in the sense that $\ell(f)\to L$ in the large-volume limit for $f$.

Now KMS states are well known to have the \emph{passivity} property \cite{PuszWoronowicz:passive}, which here amounts to
\begin{equation}
     ( \pi(A) \Omega, L \pi(A) \Omega) \geq 0 \quad \text{for all (sufficiently regular) $A =A^\ast\in \A$}.
\end{equation}
Thus the Liouvillian $L$ is positive in expectation values for states generated from $\Omega$ by operations in the thermal system, rather than in the thermal bath. It is then reasonable to expect that $\ell(f)$ shows a similar almost-positivity.
In fact, setting $\vpsi = \pi(A) \Omega$, we have
\begin{equation}
     (\vpsi, \ell(f) \vpsi) = (\vpsi, h(f) \vpsi) - (\pi(A^\ast A) \Omega, \tilde h(f) \Omega),
\end{equation}
where in the second summand we used that $\tilde h(f)$ commutes with $\pi(A)$. Supposing that $h(f)$ fulfills a QEI of the usual form \eqref{eq:qei0}, this leads to the estimate
\begin{equation}\label{eq:est0}
     (\vpsi, \ell(f) \vpsi) \geq -c_f \|\vpsi\|^2 - \| \tilde h(f) \Omega \| \, \|  \pi(A^\ast A) \Omega \| .
\end{equation}
This would be a QEI for $\ell$ of the desired form \emph{if} we could find a suitable norm $\| \cdot \|'$ such that
\begin{equation}\label{eq:altnorm}
\| \pi(A^\ast A) \Omega \| \leq (\| \pi(A)\Omega\|')^2 \quad \text{for all $A \in \A$}.
\end{equation}
But the usual Hilbert space norm does not have this property; and the choice $\| \pi(A)\Omega \|'=\|A\|$ would be too strict, as in applications we would like to generalize to states generated by (unbounded) quantum fields from $\Omega$.

However, with a small modification the argument can be saved. Instead of \eqref{eq:est0}, let us consider
\begin{equation}
     (\vpsi, \ell(f) \vpsi) \geq -c_f \|\vpsi\|^2 - \| e^{\beta L / 4}\tilde h(f) \Omega \| \, \| e^{-\beta L / 4} \pi(A^\ast A) \Omega \| .
\end{equation}
\sloppy
Here the norm $\| e^{\beta L / 4}\tilde h(f) \Omega \|$ turns out to be finite; % due to modular theory; 
and the expression $\| e^{-\beta L / 4} \pi(A^\ast A) \Omega \|^{1/2}$ is indeed a norm of $A \Omega$, known as the noncommutative $L^4$ norm in the theory of noncommutative $L^p$ spaces for type~III von Neumann algebras (see e.g.~\cite{ArakiMasuda}, or \autoref{app:lp} for a brief review). Thus our main claim can be formulated as follows: \emph{In typical situations, the Liouvillian density fulfills a quantum $L^4$ inequality.}
\fussy

The theory of noncommutative $L^p$-spaces is intimately linked with Tomita-Takesaki theory for von Neumann algebras, and indeed the entire setting above may be reformulated in these terms: Given a von Neumann algebra $\M$, every (cyclic and separating) state vector $\Omega$ has the KMS property with respect to the modular group of $(\M,\Omega)$; and vice versa in every KMS state, the modular group is linked to the time evolution by $\Delta^{it}=U(-\beta t)$, or $L=-\frac{1}{\beta}\log\Delta$. Thus our results apply generally to ``local densities'' of modular generators, if these exist in examples.

While the discussion so far was short and somewhat heuristic, we will detail our precise assumptions and arguments in \autoref{sec:qilp} and \autoref{sec:tqft}. We then present two applications of the general theory -- not only to show applicability of the assumptions made in the general setting, but crucially, to verify the nontriviality of the inequalities obtained.

Our first example will be the KMS state of a real scalar free field on Minkowski space (\autoref{sec:freescalar}); this is indicative of the general situation in linear thermal QFT, and in particular illuminates the role of ``particles'' vs. ``holes'' for the quantum inequalities.

The second example has a closer link to Tomita-Takesaki theory: it is the thermal state of the Rindler wedge (for a real scalar free field), see \autoref{sec:rindler}. That is, our algebra $\pi(\A)$ is the Weyl algebra of the right wedge, whereas $\tilde\pi(\A)$ is the Weyl algebra of the left wedge; the vacuum vector $\Omega$, cyclic and separating for both, is well known to be a KMS state for the group of boosts along the wedge (as it coincides with the modular group), and hence it has thermal properties (``Unruh temperature''). We show that the local density of the boosts can be considered a ``Liouvillian density'' in the sense described above, and establish a nontrivial quantum $L^4$ inequality for it.

To summarize the structure of the article: We introduce our precise notion of quantum inequalities, including $L^p$ inequalities, in \autoref{sec:qilp}. This is then applied to thermal quantum field theory in \autoref{sec:tqft}, using a model-independent setting, in which we also clarify relations between the energy density $h$ and the Liouvillian density $\ell$. \autoref{sec:freescalar} discusses the example of the thermal states of a real scalar free field, while \autoref{sec:rindler} shows that the Unruh temperature state fits into our setting. We conclude in \autoref{sec:conclusions}. Two appendices review and extend relevant background material: \autoref{app:affil} on operators affiliated with a von Neumann algebra, and \autoref{app:lp} on noncommutative $L^p$ spaces.

This paper is partially based on the Ph.D.~thesis of one of the authors \cite{Sangaletti:thesis}.

\section{Quantum \texorpdfstring{$L^p$}{Lp} inequalities}\label{sec:qilp}

We start by giving a precise formulation of the notion of quantum inequalities that we will use throughout the paper.

Let $\mathcal{H}$ be a Hilbert space with scalar product $(\cdot, \cdot)$ and norm $\|\cdot\|$, and let $\Ban\subset \mathcal{H}$ be a densely and continuously embedded Banach space with norm $\|\cdot\|_\Ban$. 
A hermitean sesquilinear form $T$ defined on a linear subspace $\mathcal{D}\subset \Ban$, with $\mathcal{D}$ dense in $\Hil$, is called \emph{$\Ban$-bounded below} if there is a $c\in\R$ such that
  \begin{equation}
     ( \vpsi, T \vpsi ) \geq  c \| \vpsi \|^2_\Ban \quad \text{for all }\vpsi\in \mathcal{D};
  \end{equation}
  \emph{$\Ban$-bounded above} if the same holds with $\leq$ instead of $\geq$; and \emph{$\Ban$-bounded} if it has both of these properties. We can now state:

\begin{definition}
Let $\mathcal{D}\subset \Ban$ be a linear subspace, dense in $\Hil$, and let $T: \mathcal{D} \times \mathcal{D} \to \mathbb{C}$ be a hermitean sesquilinear form. We say that $T$ fulfills a \emph{nontrivial quantum $\Ban$-inequality} if $T$ is $\Ban$-bounded below, but not $\Ban$-bounded above. 
\end{definition}

When $\Ban=\Hil$, this is the usual notion of a (state-independent) quantum inequality; when  
$\|\cdot\|_\Ban$ is stricter than $\|\cdot\|$, it is a state-dependent quantum inequality in the sense of \cite{FewsterOsterbrink2008}. We stress that the nontriviality condition (namely, that  $T$ is \emph{not} $\Ban$-bounded above) is crucial in our applications; it expresses the fact that $T$ is ``more positive than negative'' with respect to the norm of $\Ban$. 

More specifically, we now consider the following situation. Let $\M$ be a von Neumann algebra acting on $\Hil$, with commutant $\M'$, and let $\Omega\in\Hil$ be a cyclic and separating vector for $\M$. We denote the Tomita-Takesaki modular objects as $\Delta$, $J$. 

Given this, as our dense subspaces of $\Hil$ we consider the \emph{noncommutative $L^p$ spaces} $L^p(\M,\Omega)$, $2 \leq p \leq \infty$; see \autoref{app:lp} for a brief review of these. In short, these $L^p$ spaces can be thought of as a chain of densely and continuously embedded Banach spaces, 
$\Hil \supset L^p(\M,\Omega) \supset L^{p'}(\M,\Omega) \supset L^\infty(\M,\Omega)$ for $2 \leq p \leq p' \leq \infty$, with $\M\Omega$ being a dense subset of each of them. The corresponding norms $\|\cdot \|_p$ arise from an abstract interpolation construction; but for certain $p$, they have more concrete expressions, in particular one has for any $A\in\M$,
\begin{equation}\label{eq:noncnorms}
   \|A\Omega\|_2 = \|A\Omega\|, \quad 
   \|A\Omega\|_4 = \|\Delta^{1/4} A^\ast A\Omega\|^{1/2}, \quad 
   \|A\Omega\|_\infty = \|A\|_\mathrm{op}. 
\end{equation}
The first and the last identities are well known and allow us to identify $L^2(\M,\Omega)=\Hil$, $L^\infty(\M,\Omega)=\M\Omega \cong \M$. The expression for $\|A\Omega\|_4$ seems less known, but can be extracted from the literature (e.g.,~Theorem~3.6 of \cite{Raynaud:lp}); we report a derivation from \cite{ArakiMasuda} in \autoref{app:lp}.
While \eqref{eq:noncnorms} was stated for bounded $A\in\M$, these formulas extend -- under suitable conditions  -- to operators $A$ which are unbounded, but affiliated with $\M$; see Proposition~\ref{prop:affilp4norm} in the appendix.

Our main interest in the following will be in (semi-)boundedness with respect to $\Ban=L^4(\M,\Omega)$. For use later on, we prove a number of lemmas regarding $L^4$-boundedness.

\begin{lemma}\label{lemma:tpbound}
    Let $T'$ be a densely defined closed symmetric operator affiliated with $\M'$ such that $\Omega\in\dom{T'}$. The sesquilinear form associated with $T'$ and defined on $\M\Omega\times\M\Omega$ is $L^4(\M,\Omega)$-bounded.
\end{lemma}
\begin{proof}  
    Since $T'\affil\M'$ and $\Omega\in\dom{T'}$, we have $\M\Omega\subset\dom{T'}$. The sesquilinear form associated with $T'$ defined on $\M\Omega\times\M\Omega$ satisfies (for every $A\in\M$):
     \begin{equation}\label{Eq:L4_abstract_ineq}
    \left(A\Omega,T'A\Omega\right)=\left(A\Omega,AT'\Omega\right)=\left(A^*A\Omega,T'\Omega\right),
    \end{equation}
since $T'$ is affiliated with $\M'$. Now $A^*A\in\M$ implies $A^\ast A \in \dom{\Delta^{1/4}}$; and $T' \affil \M'$ with $\Omega\in\dom{T'}\cap \dom{T^{\prime\ast}}$ implies $T'\Omega \in \dom{\Delta^{-1/2}}\subset \dom{\Delta^{-1/4}}$ by  Proposition \ref{prop:Ext_Tomita} (applied to $\M'$ in place of  $\M$). Thus we have
     \begin{equation}
     \left(A\Omega,T'A\Omega\right)=
     \left(\Delta^{1/4}A^*A\Omega,\Delta^{-1/4}T'\Omega\right).
    \end{equation}
Using the Cauchy-Schwarz inequality, we finally get: 
\begin{equation}
     \left|\left(A\Omega,T'A\Omega\right)\right|\leq \left\lVert\Delta^{1/4}A^*A\Omega\right\rVert \left\lVert\Delta^{-1/4}T'\Omega\right\rVert=c^2_{T'}\|A\Omega\|^2_4
\end{equation}
with a constant $c_{T'}$ independent of $A$.
    \end{proof}

In the previous lemma, the domain $\M\Omega$ can be enlarged to contain vectors generated from $\Omega$ with \emph{unbounded} operators affiliated with $\M$. More precisely:

\begin{lemma} \label{lem:extended_l4_set}
    Let $T'$ be a densely defined closed symmetric operator affiliated with $\M'$ such that $\Omega\in\dom{T'}$. Set
    \begin{equation*}
        \domAffil{T'} \coloneqq \left\{ A\Omega\; \middle\vert \; A\affil\M,\;\Omega\in\dom{A^*A},\;A\Omega\in\dom{T'}\right\}.
    \end{equation*}
    Then $\domAffil{T'} \subset L^4(\M,\Omega)$; and if $\mathcal{D}\subset\domAffil{T'}$ is a dense linear subspace\footnote{We do not claim that the subset $\domAffil{T'}\subset\Hil$ is, in general, a linear subspace.} of $\Hil$, then the sesquilinear form associated with $T'$ and defined on $\mathcal{D} \times \mathcal{D}$ is $L^4(\M,\Omega)$-bounded.
\end{lemma}
    \begin{proof}
        First, $\domAffil{T'} \subset L^4(\M,\Omega)$ follows from Proposition~\ref{prop:affilp4norm} in the appendix. Now we follow similar lines as for Lemma~\ref{lemma:tpbound}. In particular, for every vector $A\Omega\in\domAffil{T'}$ we have:
        \begin{equation}
\left(A\Omega,T'A\Omega\right)=(A\Omega,AT'\Omega)=\left(A^*A\Omega,T'\Omega\right)=\left(\Delta^{1/4}A^*A\Omega,\Delta^{-1/4}T'\Omega\right).
        \end{equation}
        In this case, the first equality is obtained applying Lemma \ref{lem: affaff}.
        The third equality, instead, makes use of Proposition \ref{prop:Ext_Tomita} to ensure the condition $A^*A\Omega\in \dom{\Delta^{1/4}}$ (recalling that, if $A\affil\M$, also $A^*A\affil\M$ by the polar decomposition theorem). Now the Cauchy-Schwarz inequality and Proposition \ref{prop:affilp4norm} lead us to
        \begin{equation}
            \left|\left(A\Omega,T'A\Omega\right)\right|\leq c^2_{T'}\|A\Omega\|_4^2
        \end{equation}
        as claimed.
    \end{proof}
The generalisation in Lemma~\ref{lem:extended_l4_set} will turn out to be crucial since in our applications to quantum energy inequalities, the operators $A$ will be unbounded quantum fields; see \autoref{sec:nontriv} below. Lemma~\ref{lem:extended_l4_set} allows us to extend the $L^4$ bounds to this case, whereas $L^\infty$ bounds (operator norm bounds) would clearly not hold in that situation.

\section{Thermal field theory and the density of the Liouvillian}\label{sec:tqft}

We now introduce the setting of thermal field theory that we shall use in the following. Essentially, we make use of the framework for thermodynamics in infinite volume put forward by Haag, Hugenholtz and Winnink \cite{HHW:equilibrium}. Let us fix the details, mostly following \cite{BraRob:qsm2}; see also \cite{Haa:LQP} for context.

Let $\spt$ be a static, globally hyperbolic spacetime,\footnote{In our examples, $\spt$ will be 3+1-dimensional Minkowski space or a part of it; but our methods work in the more abstract setting.} on which we fix a Cauchy surface $\csf$ and the flow $\Xi$ of a timelike Killing vector orthogonal to $\csf$ (i.e.,~the ``time evolution''). We will take subregions $\ocal \subset \spt$ to be open, with (the interior of) their causal complement denoted as $\ocal'$. If $\ocal''=\ocal$, we call $\ocal$ \emph{causally complete}.

A quantum field theory on the spacetime $\spt$ is given by a net of $C^\ast$-algebras, $\A:\ocal \mapsto \A(\ocal)$, defined on the open subregions $\ocal$ of $\spt$, which fulfills the conditions of isotony ($\A(\ocal_1)\subset \A(\ocal_2)$ if $\ocal_1\subset\ocal_2$) and locality ($[\A(\ocal),\A(\ocal')]=\{0\}$). For technical convenience, we also demand that $\A$ is inner regular ($\A(\ocal)=\overline{\bigcup_{\ocal_1\subset\subset\ocal} \A(\ocal_1)}^{\lVert \cdot \rVert}$, where $\ocal_1\subset\subset\ocal$ stands for $\ocal_1\subset \mathcal{K} \subset \ocal$ with a compact set $\mathcal{K}$). 
%The quasi-local algebra is $\A(\spt)$. 
We further require that the net is covariant with respect to the flow $\Xi$, i.e., that there exists a one-parameter group of automorphisms $\alpha_t:\A(\spt)\to\A(\spt)$ such that $\alpha_t (\A (\ocal)) = \A(\Xi_t \ocal)$ for all $\ocal$.

We then consider invariant states $\omega$ on $\A$, that is, positive linear normalized functionals $\omega:\A(\spt)\to\mathbb{C}$ such that $\omega \circ \alpha_t = \omega$.
Constructing the GNS representation $\pi$ for $(\A(\spt),\omega)$, we find a Hilbert space $\Hil$ and a vector $\Omega \in \Hil$ such that $\omega(A) = (\Omega,\pi(A)\Omega)$ for all $A$. We also obtain a one-parameter unitary group $U(t)$ such that $\pi(\alpha_t(A)) = U(t)\pi(A)U(t)^\ast$. We shall assume\footnote{With this assumption, we avoid imposing continuity requirements on the automorphisms $\alpha_t$ with respect to $t$, while still retaining the consequences of Corollary~5.3.4 in \cite{BraRob:qsm2} needed later.} that $U(t)$ is strongly continuous, hence $U(t)=e^{iLt}$ with a self-adjoint operator $L$, the \emph{Liouvillian}.

We thus obtain an isotonous, local, $U$-covariant net of von Neumann algebras $\M:\ocal \mapsto \M(\ocal) \coloneqq \pi(\A(\ocal))''$ on $\Hil$; and the algebra $\M(\spt)$ will take the role of the von Neumann algebra $\M$ from the previous section. By construction, $\Omega$ is cyclic for $\M(\spt)$.

We are particularly interested in states describing thermal equilibrium; technically, those fulfilling the KMS property with respect to $\alpha$ for some inverse temperature $0<\beta<\infty$, in short,
\begin{equation}
    \omega( A \alpha_{t+i\beta}(B) ) = \omega( \alpha_t(B) A )
    \quad \text{for all }
    t\in\R,\; A,B\in\A(\spt).
\end{equation}
As a consequence of this condition, $\Omega$ is also separating for $\M(\spt)$ \cite[Corollary~5.3.9]{BraRob:qsm2}. Hence we can construct the associated Tomita-Takesaki modular objects $\Delta$, $J$. By uniqueness of the modular automorphism group \cite[Theorem~5.3.10 and its proof]{BraRob:qsm2}, we obtain
\begin{equation}
    \Delta^{it} = e^{-it \beta L}, \quad \text{or} \quad
    L = - \frac{1}{\beta} \log \Delta .
\end{equation} 
The modular objects allow us to define a second (anti-linear) representation $\tilde \pi$ of $\A(\spt)$  by $\tilde \pi(A) := J\pi(A) J$. Defining $\Mt (\ocal):= J\M(\ocal) J = \tilde \pi(\A(\ocal))''$, we obtain another isotonous, local, covariant net; $\M$ and $\Mt$ are interpreted as the nets of observable algebras of the thermal system and the thermal bath respectively. They commute with each other; indeed, we have $\M(\spt)=\Mt(\spt)'$. 
We further set
\begin{equation}
    \Mh(\ocal) \coloneqq \M(\ocal) \vee \Mt(\ocal),
%    \quad
%    \Mh_\dual(\ocal) = \M_\dual(\ocal) \vee \Mt_\dual(\ocal).
\end{equation}
giving another isotonous, local, $U(t)$-covariant net that describes the system \emph{combined} with the bath.

We now wish to describe a \emph{local density} of the Liouvillian, similar to an energy density in the vacuum case. Heuristically, such a density $\ell$ should be given as a quantum field of the form \cite[Eq.~(1.5)]{HHW:equilibrium}
\begin{equation}\label{eq:lfromh}
    \ell(x) = h(x) - \tilde h(x), \quad x \in \spt,
\end{equation}
where $h$ is the local energy density associated with $\M$ (i.e., the thermal system) and $\tilde h$ the one associated with $\Mt$ (i.e., the thermal bath). In order to formalize this mathematically, we need to describe (at least) three different aspects of $\ell$:
\begin{enumerate}[(1)]
 \item that it is a covariant quantum field, or technically an operator-valued distribution, associated with the net $\Mh$,
 \item \label{it:twoparts} that it splits into two parts, one associated with $\M$ and one with $\Mt$,
 \item that, when integrated over the Cauchy surface $\csf$, it yields the Liouvillian $L$.
\end{enumerate}
We will treat these aspects one by one, and draw consequences from our definitions. Our main point is that (\ref{it:twoparts}), together with a ``usual'' QEI for $h$,  leads to a quantum $L^4$ inequality for $\ell$.

\subsection{Relation to algebras}

We base our notion of a local quantum field on affiliation with a net of von Neumann algebras, cf.~\cite{FreHer:pointlike_fields}. More formally:
\begin{definition}\label{def:associated}
 Let $\N$ be an isotonous and $U$-covariant net of von Neumann algebras on $\Hil$. A \emph{quantum field $\phi$ associated with $\N$} is a complex-linear map $\phi: f \mapsto \phi(f)$ which assigns to every test function $f \in \ccs(\spt)$ a densely defined closed operator $\phi(f)$ on $\Hil$, such that  $\phi(f) \affil \N(\ocal')'\cap\N(\spt)$ whenever $\operatorname{supp} f \subset \ocal$. We call $\phi$ \emph{covariant} if $\phi( \Xi_t^\ast f) = U(t) \phi(f) U(t)^\ast$ for all $f,t$. We call $\phi$ \emph{symmetric} if $\phi(f)$ is symmetric whenever $f$ is real-valued.  
\end{definition}
By complex-linear, we mean that $\phi(\lambda f+g)=\lambda\phi(f)+\phi(g)$ on a common core for all three operators. Note that we do not require a common dense domain for all $\phi(f)$, nor that $\Omega \in \dom{\phi(f)}$ a priori. We also do not impose continuity conditions  on $\phi(f)$ in $f$.

Somewhat against usual conventions, the locality of our fields is defined by relative commutation to the net $\N$, not by commutation to itself at spacelike distances;%
\footnote{%
In usual terms, we demand that our fields are affiliated with the algebras of the \emph{dualized} net, $\N_\dual(\ocal)\coloneqq\N(\ocal')'\cap\N(\spt)$. This net $\N_\dual$ is isotonous and covariant, but need not be local even if $\N$ is. In quantum field theory, at least in the case of the vacuum representation $\pi$ of simple models and for sufficiently regular regions $\ocal$, it is known that $\pi(\A(\ocal))'' = \pi(\A(\ocal'))'$ (implied by \emph{Haag duality}), see, e.g., \cite{Ara:lattice}; thus $\M=\M_\dual$. By analogy, one may conjecture that in thermal representations of such models, we have $\M=\M_\dual$ as well, or at least that $\M_\dual$ is in itself a local net, a variant of \emph{essential duality} \cite{Yngvason:duality}. However, to the best of the authors' knowledge, this model-dependent question has not been settled in the literature for the thermal case, and it is beyond the scope of our current investigation. 
}
this is technically better suited when working with affiliated operators, see Appendix~\ref{app:affil}.

For our applications, the net $\N$ will be one of $\M$, $\Mt$ and $\Mh$. Specifically, the Liouvillian density $\ell$ should be a covariant symmetric quantum field associated with $\Mh$.

\subsection{Splitting}

The next aspect is a split of the field into a ``system'' and a ``bath'' part. Formally we define:
\begin{definition}\label{def:split}
Let $\ell$ be a symmetric quantum field associated with $\Mh$. 
We say that $\ell$ is \emph{split} if there are symmetric quantum fields $h$, $\tilde h$ associated with $\M$, $\Mt$ respectively, and a dense subspace $\domSplit\subset \Hil$ with $\domSplit \subset \dom{\ell(f)} \cap \dom{h(f)} \cap \dom{\tilde h(f)}$ for all $f$, such that
 \begin{equation}\label{eq:lsplit}
    (\vpsi, \ell(f)  \vpsi) = (\vpsi, h(f) \vpsi)  - (\vpsi, \tilde h(f) \vpsi)
 \end{equation}
 for all $f \in \ccs(\spt)$ and $\vpsi\in\domSplit$.
\end{definition}

When in applications, two fields $h, \tilde h$ are given, then under suitable conditions on the operator domains, their sum or difference defines a split field  in the above sense:

\begin{proposition} \label{prop:ellsum}
    If $h$, $\tilde h$ are symmetric quantum fields associated with $\M,\Mt$ respectively, and if all $h(f)$, $\tilde h(f)$ have a common core $\mathcal{D}$, then
    \begin{equation*}
        \ell(f) := \overline{h(f) - \tilde h(f)}
    \end{equation*}
    defines a symmetric quantum field $\ell$ associated with $\Mh$, which is split with $\domSplit=\mathcal{D}$. 
\end{proposition}
\begin{proof}
 The hypothesis guarantees that $h(f)-\tilde h(f)$ is densely defined%
 %($\mathcal{D}\subset \dom{h(f)}\cap \dom{\tilde h(f)}$)
 ; it is also closable since it is symmetric. It inherits linearity in $f$ from $h$ and $\tilde h$, with $\mathcal{D}$ providing the required common core. Further, with $h(f)\affil\M(\ocal')'\cap \M(\spt)$ and $\tilde h(f)\affil\Mt(\ocal')'\cap \Mt(\spt)$, Lemma~\ref{lem:sumaffil} in the appendix shows that $\ell(f)$ is affiliated with
 \begin{equation}\label{eq:affilunion}
   \big(\M(\ocal')'\wedge \M(\spt) \big) \vee \big( \Mt(\ocal')'\wedge \Mt(\spt) \big)
   = 
   \big(\M(\ocal') \vee \Mt(\ocal')\big)' 
   = \Mh(\ocal')'
 \end{equation}
 as desired. It is clear by construction that $\ell$ is split. 
\end{proof}

Now if a quantum field is split, our results in \autoref{sec:qilp} imply that a quantum $L^4$  inequality for $h$ -- which would follow from a quantum energy inequality in the ``usual'' sense -- implies the same type of inequality for $\ell$.

\begin{theorem}\label{thm:splitdensity}
    Let $\ell=h-\tilde h$ be a symmetric and split quantum field, 
    and let $f \in \ccs(\R^4)$ be such that $\Omega\in\domSplit$. Let $\mathcal{D}\subset \Hil$ be a dense subspace with $\mathcal{D} \subset \domSplit \cap \domAffil{\tilde h(f)}$. Then $\ell(f)$, as a sesquilinear form on $\mathcal{D} \times \mathcal{D}$, fulfills a nontrivial quantum $L^4(\M(\spt),\Omega)$ inequality if and only if $h(f)$ does.
\end{theorem}
\begin{proof}
Since $\tilde h(f)\affil \Mt(\spt)=\M(\spt)'$, Lemma~\ref{lem:extended_l4_set} shows that $\tilde h(f)$ is $L^4(\M(\spt),\Omega)$-bounded on $\mathcal{D}$ by a constant $\tilde c>0$. From \eqref{eq:lsplit},  the triangle inequality yields for all $\vpsi \in \mathcal{D}$,
\begin{equation}
     \left(\vpsi,h(f)\vpsi\right) -  \tilde c \|\vpsi\|_4^2
     \leq 
    \left(\vpsi,\ell(f)\vpsi\right)
     \leq
     \left(\vpsi,h(f)\vpsi\right) +  \tilde c \|\vpsi\|_4^2.
\end{equation}
Thus a (nontrivial) quantum $L^4$ inequality for $h(f)$ on $\mathcal{D}\times\mathcal{D}$ implies a (nontrivial) quantum $L^4$ inequality for $\ell(f)$ there and vice versa.
\end{proof}

\subsection{Integration to the Liouvillian}

Lastly, we ask in which sense the field $\ell$ is related to the Liouvillian. Already for the energy density in the vacuum sector on Minkowski space, this question is mathematically quite delicate, and normally treated by means of test functions $f$ that are progressively stretched out along a Cauchy surface \cite{Orzalesi:currents}. We choose a more direct approach here, referring to ``pointlike fields'' as sesquilinear forms, as this will be the easiest to verify in examples.

\begin{definition}\label{def:integration}
  We say that a covariant symmetric quantum field $\ell$ associated with a net $\N$ is a \emph{density for $L$ in the strong sense} if there exists a dense subspace $\mathcal{D}\subset \Hil$, invariant under $U(t)$, contained in $\dom{L}$ and being a core for every $\ell(f)$; and for every $x\in\spt$ there exists a sesquilinear form $\ell(x)$ on $\mathcal{D} \times \mathcal{D}$, such that for every $\vpsi \in \mathcal{D}$,
  \begin{enumerate}[(i)]
%   \item $\vpsi \in  \dom{L}$, and $\mathcal{D}$ is a core\todo{condition strengthened} for every $\ell(f),$ $f \in \ccs(\spt)$;
   \item \label{it:fieldcont} the function $\spt \to \mathbb{C}$, $x \mapsto (\vpsi, \ell(x) \vpsi)$ is continuous;
   \item \label{it:fieldint}$\displaystyle (\vpsi, \ell(f) \vpsi) = \int_{\spt} f(x)\, \big(\vpsi, \ell(x) \vpsi\big)\, \mathrm{dVol}_{\spt}(x)$ for any $f \in \ccs(\spt)$;
   \item  \label{it:ellint} $\displaystyle(\vpsi, L \vpsi) = \int_{\csf} \big(\vpsi, \ell(x) \vpsi\big) \,\mathrm{dVol}_{\csf}(x)$.
  \end{enumerate}
\end{definition}

Note that by covariance of the field and invariance of $\mathcal{D}$ and $L$ under $U(t)$, condition (\ref{it:ellint}) implies that
\begin{equation}
\label{eq:ellintcov} 
\begin{aligned}(\vpsi, L \vpsi) &= (U(t)\vpsi, L U(t)\vpsi) 
\\&= \int_{\csf} \big(\vpsi, \ell(\Xi_t(x)) \vpsi\big) \,\mathrm{dVol}_{\csf}(x) \quad \text{for every }\vpsi\in\mathcal{D},\, t \in \R.
\end{aligned}
\end{equation}

We also define an alternative, much weaker notion of density for the generator, based on commutators with the local algebras, in the spirit of \cite{Verch:anec}. This is mainly done to justify the formula \eqref{eq:lfromh} for the Liouvillian density at a model-independent level. The approach is to define a Liouvillian density as a field that, when averaged in space-time (but in a suitable neighbourhood of the time-0 base of a given region $\ocal$), gives the same commutator with local observables in $\ocal$ as the full Liouvillian $L$.  

In preparation, we first introduce a special class of test functions.

\begin{definition}\label{def:spacelikeone}
   Let $\ocal \subset \ocal_1$ be two regions in $\spt$. We say that $f \in \ccs(\ocal_1)$ is \emph{spacelike-one for $\ocal$} if it is of the form
   \begin{equation*}
       f( \Xi_t x_0 ) = f_T(t) f_S(x_0), \quad x_0 \in \csf, \; t \in \R,
   \end{equation*}
   with some $f_T\in\ccs(\R)$, $f_S \in \ccs(\csf)$, where $f_T\geq 0$, $\int f_T(t) \,\mathrm{d}t = 1$, and for each $x_0\in\csf$ either $f_S=1$ in  a neighbourhood of $x_0$, or $\Xi_t x_0 \in \ocal'$ for all $t \in \operatorname{supp}(f_T)$.
\end{definition}

Also, for any net $\N$ of von Neumann algebras, the Liouvillian introduces a symmetric spatial derivation on the algebra $\N(\spt)$ by $\delta(A)=i[L,A]$, see \cite[Sec.~3.2.5]{BraRob:qsm1}, on a uniquely defined domain $\dom{\delta}\subset \N(\spt)$, cf.~\cite[Proposition~3.2.55]{BraRob:qsm1}. The derivation is closable in the $\sigma$-weak topology. Elements in its domain can be approximated as follows: If for $A\in\mathfrak{B}(\Hil)$, $f \in \ccs(\R)$ we set
\begin{equation}\label{eq:Bint}
  B = A(f) \coloneqq \int f(t) U(t)AU(t)^\ast \,\mathrm{d}t
\end{equation}
as a weak integral, then $B \in \dom{\delta}$, and  $\delta(B)=A(-f')$; in fact, $B \in \dom{\delta^n}$ for any $n$ (we call such elements \emph{smooth}). By taking $f$ as a delta sequence, we can $\sigma$-weakly approximate any $A$ with elements of the form \eqref{eq:Bint}; and if $A\in\dom{\delta}$, then we can thus approximate $(A,\delta(A))$ in the graph sense. If $\N$ is inner regular, such approximations can be done within every local algebra $\N(\ocal)$, i.e., elements of the form $A(f)$, $A\in\N(\ocal)$, are $\sigma$-weakly dense in $\N(\ocal)$. 

We can now state our weaker notion for the density of generators.
\begin{definition} \label{def:densityweak}
   A covariant net $\N$ is said to have \emph{local densities  in the weak sense} for the generator $L$ if, for every causally complete region $\mathcal{O}\subset \spt$ and every precompact region $\ocal_1 \supset \bar\ocal$ with nonempty $ \ocal_1'$, there is a symmetric quantum field $\ell_\ocal$ on the spacetime $\ocal_1$ associated with $\N \restriction \ocal_1$ (the local density) with the following property. For every $f \in \ccs(\ocal_1)$ which is spacelike-one for $\ocal$, every $A \in \N(\ocal) \cap \dom{\delta}$,  and every $\vpsi \in \dom{\ell_\ocal(f)}$, it holds that
   \begin{equation}\label{eq:locomm}
      (\vpsi, \delta(A) \vpsi) = i (\vpsi, [ \ell_\ocal(f), A ] \vpsi ).
   \end{equation}
\end{definition}

The commutator in \eqref{eq:locomm} is meant in the weak sense. Note that the field $\ell_\ocal$ does not only depend on $\ocal$, but also on the choice of $\ocal_1$, but we suppress this in the notation. 

We show that the above definition is consistent with our notion of generators in the \emph{strong} sense:  

\begin{proposition}
   If $\ell$ is a quantum field associated with a net $\N$ which is a density for $L$ in the strong sense, and $\mathcal{D}$ is invariant under smooth elements of $\N(\ocal)$ for precompact $\ocal$, then $\N$ has local densities in the weak sense, namely $\ell_{\ocal} = \ell \restriction \ccs(\ocal_1)$.
\end{proposition}
\begin{proof}
    Fix $\ocal,\ocal_1$ as in Definition~\ref{def:densityweak} and let $f=f_Sf_T\in\ccs(\ocal_1)$ be spacelike-one for $\ocal$. From Eq.~\eqref{eq:ellintcov} and since $\int f_T(t)\di t=1$, we have 
    \begin{equation}
       (\vpsi, L\vpsi) = \int f_T(t) \big(\vpsi,\ell(\Xi_t(x_0))\vpsi\big) \,\di t \,\mathrm{dVol}_{\csf}(x_0).
    \end{equation}
    Choosing a suitable sequence $g_n\in\ccs(\csf)$ which is 1 on an increasing sequence of compact sets exhausting $\csf$, we can write for $\vpsi\in\mathcal{D}$,
    \begin{equation}
    \begin{aligned}
       (\vpsi, L\vpsi) &= \lim_{n \to \infty}\int \Big(f_T(t)f_S(x_0) + \underbrace{f_T(t)\big(1-f_S(x_0)\big)g_n(x_0)}_{=:h_n(\Xi_t(x_0))} \Big) \big(\vpsi,\ell(\Xi_t(x_0))\vpsi\big) \,\di t \,\mathrm{dVol}_{\csf}(x_0)
       \\
       &= \lim_{n \to \infty}\int (f(x) + h_n(x)) \,(\vpsi,\ell(x)\vpsi)\, \mathrm{dVol}_{\spt}(x).
    \end{aligned}
    \end{equation}
    Here $h_n \in \ccs(\ocal')$, since $f$ is spacelike-one for $\ocal$. By polarization and using property (\ref{it:fieldint}) in Def.~\ref{def:densityweak}, we obtain for any $\vpsi,\vpsi'\in\mathcal{D}$,
    \begin{equation}
      (\vpsi, L \vpsi') = \lim_{n \to \infty}(\vpsi, \big(\ell(f)+\ell(h_n)\big) \vpsi').
    \end{equation}
    If now $A$ is a smooth element of $\N(\ocal)$, and thus $A \mathcal{D} \cup A^\ast \mathcal{D}\subset\mathcal{D}$, we obtain 
    \begin{equation}
      (\vpsi, \delta(A) \vpsi) = i(\vpsi, [L,A] \vpsi) = i (\vpsi, [\ell(f),A] \vpsi) + i \lim_{n\to\infty}(\vpsi, [\ell(h_n), A ] \vpsi).
    \end{equation}
    But since $\ell(h_n) \affil \N(\ocal'')'\cap \N(\spt)\subset \N(\ocal)'$ (here causal completeness of $\ocal$ enters), the second summand vanishes. That is, we have obtained the relation \eqref{eq:locomm} (with $\ell_\ocal=\ell$) for $\vpsi\in\mathcal{D}$ and any smooth $A$. With $\mathcal{D}$ being a core for $\ell(f)$, \eqref{eq:locomm} then holds for all $\vpsi\in\dom{ \ell(f)}$, and  with smooth elements being dense in $\dom{\delta}$, it holds for all $A \in \dom{\delta}$.
\end{proof}

The main point about our weaker definition is that it allows us to give meaning to the heuristic formula \eqref{eq:lfromh} for the generator. The idea is to lift a local density that exists in the vacuum sector to the thermal state, which is ``locally'' in the same folium. To that end, we now consider \emph{two} invariant states $\omega_0$, $\omega_\beta$ of the same net $\A$ of $C^\ast$-algebras, of which only $\omega_\beta$ is assumed to be a $\beta$-KMS-state. The other state, $\omega_0$, may be thought of as a vacuum state (but very few specific properties of vacuum states will be needed in the following). Hence we have two GNS representations, $\pi_0$ on $\Hil_0$ and $\pi_\beta$ on $\Hil_\beta$, and we will label the corresponding von Neumann algebras, generators, etc.~accordingly. For consistency, we also label the modular data for $(\M_\beta(\spt),\Omega_\beta)$ as $J_\beta$, $\Delta_\beta$; those for $(\M_0(\spt),\Omega_0)$ need not exist.

The two representations of $\A$ are said to be \emph{locally normal} to each other if, for every precompact $\ocal$ with nonempty $\ocal'$, the representations $\pi_0 \restriction \A(\ocal)$ and $\pi_\beta \restriction \A(\ocal)$ are quasi-equivalent. This means (cf.~\cite[Def.~10.3.1]{KadRin:algebras2}) that there is a $\ast$-isomorphism $\varphi_\ocal$ between $\M_0(\ocal)$ and $\M_\beta(\ocal)$ such that $\varphi_\ocal \circ \pi_0 \restriction \A(\ocal)= \pi_\beta \restriction \A(\ocal)$. Let us assume that both $\M_0$ and $\M_\beta$ have the Reeh-Schlieder property; i.e., $\Omega_0$ is cyclic for $\M_0(\ocal)$ for every nonempty $\ocal$, similarly for $\M_\beta$. (This is a realistic assumption also in the thermal case \cite{Jaekel:rsThermal}.) Replacing $\ocal$ with $\ocal'$ (if nonempty), it follows that $\Omega_0$, resp.\ $\Omega_\beta$, is also separating. As a consequence, the isomorphism $\varphi_\ocal$ is unitarily implementable \cite[Theorem~7.2.9]{KadRin:algebras2}, i.e., there is a unitary $V_\ocal : \Hil_0 \to \Hil_\beta$ such that
\begin{equation}\label{eq:vpv}
    V_\ocal \pi_0(A) V_\ocal^\ast = \pi_\beta(A) \quad \text{ for all $A \in \A(\ocal)$}.
\end{equation}
A forteriori, for any subregion $\ocal_<\subset\ocal$, we have
\begin{equation}\label{eq:vMv}
    V_\ocal \M_0(\ocal_<) V_\ocal^\ast = \M_\beta(\ocal_<),
\end{equation}
since the $\pi_{0,\beta}(A)$, $A\in\A(\ocal_<)$, form a weakly dense subalgebra of $\M_{0,\beta}(\ocal_<)$.
Correspondingly, if $\phi_0$ is a quantum field on $\ocal$ associated with $\M_0 \restriction \ocal$, then
\begin{equation}\label{eq:fieldtransf}
    \phi_\beta(f) :=  V_\ocal \phi(f) V_\ocal^\ast
\end{equation}
defines a quantum field $\phi_\beta$ associated with $\M_\beta \restriction \ocal$, where $\ocal$ now takes the role of the spacetime $\spt$; the domain of the closed operator $\phi_\beta(f)$ is $V_\ocal \dom(\phi_0(f))$.

We note consequences for the derivations $\delta_0$ on $\M_0$ and $\delta_\beta$ on $\Mh_\beta$.

\begin{lemma}\label{lemma:gentogen}
 In the above situation, one has  $V_\ocal (\dom{\delta_0} \cap \M_0(\ocal)) V_\ocal^\ast = \dom{\delta_\beta} \cap \M_\beta(\ocal)$, and 
 \begin{align} \label{eq:vdv}
     V_\ocal \delta_0(B) V_\ocal^\ast = \delta_\beta( V_\ocal B V_\ocal^\ast)\quad \text{ for all }B \in \dom{\delta_0} \cap \M_0(\ocal).
 \end{align}
 Further, it holds that
 \begin{equation}
     J_\beta \delta_\beta(A) J_\beta =  \delta_\beta(J_\beta A J_\beta)
     \quad \text{ for all $A \in \dom{\delta_\beta}$}.
 \end{equation}
\end{lemma}
\begin{proof}
   Consider $B$ of the form
   \begin{equation}\label{eq:bform}
     B=\int f(t) U_0(t)  \pi_0(A) U_0(t)^\ast \,\mathrm{d}t
  \end{equation}
  where $f \in \ccs(\R)$, $A\in \A(\ocal)$ are such that $\alpha_t(A) \in \A(\ocal)$ for all $t \in \operatorname{supp} f$. This $B$ lies in $\dom(\delta_0)$, and 
  \begin{equation}\label{eq:bsmear}
  \begin{aligned}
     V_\ocal B V_\ocal^\ast 
     &= \int f(t) V_\ocal \pi_0(\alpha_t A) V_\ocal^\ast \,\mathrm{d}t
     = \int f(t) \pi_\beta(\alpha_t A) \,\mathrm{d}t \\ &= \int f(t) U_\beta(t)\pi_\beta(A) U_\beta(t)^\ast\,\mathrm{d}t 
     = \int f(t) U_\beta(t) V_\ocal \pi_0(A) V_\ocal^\ast U_\beta(t)^\ast \,\mathrm{d}t.
     \end{aligned}
    \end{equation}
  Thus $V_\ocal B V_\ocal^\ast \in\dom{\delta_\beta}$; and from a computation similar to \eqref{eq:bsmear}, replacing $f$ with $-f'$, one sees that $\delta_\beta( V_\ocal B V_\ocal^\ast ) =  V_\ocal \delta_0(B) V_\ocal^\ast $.

  Thus \eqref{eq:vdv} holds for $B$ of the form \eqref{eq:bform}. Thanks to inner regularity, we can then $\sigma$-weakly approximate any $(B,\delta_0(B))$ in the graph of $\delta_0 \restriction \M_0(\ocal)$ with elements of the above form (see the remark after Def.~\ref{def:spacelikeone}); since $\delta_\beta$ is $\sigma$-weakly closed, this yields  \eqref{eq:vdv} for any $B \in\dom{\delta_0} \cap \M_0(\ocal)$, including the statement that $V_\ocal (\dom{\delta_0} \cap \M_0(\ocal))V_\ocal^\ast \subset \dom{\delta_\beta} \cap \M_\beta(\ocal)$.  The opposite inclusion follows by reversing the roles of $\delta_0$ and $\delta_\beta$.  
  
  The last part of the lemma follows from the known relation $J_\beta L_\beta J_\beta = - L_\beta$.
\end{proof}

We also recall the notion of a \emph{split inclusion} of von Neumann algebras: An inclusion of two von Neumann algebras $\mathcal{M} \subset \N$ on a Hilbert space $\Hil$ is called \emph{split} if there is a type-I factor $\mathcal{F}$ such that $\mathcal{M} \subset \mathcal{F} \subset \N$. Equivalently \cite[Lemma~2]{DAntoniLongo:interpolation}, there exists a Hilbert space $\mathcal{K}$ and a unitary $Y: \Hil \to \mathcal{K}\otimes\mathcal{K}$ such that $Y \mathcal{M} Y^\ast \subset \mathfrak{B}(\mathcal{K}) \otimes \idop$ and $Y \N' Y^\ast \subset \idop \otimes \mathfrak{B}(\mathcal{K})$. This split property has been discussed extensively in relation to the local algebras in quantum field theory, see \cite{Fewster:splitReview} for a review.

We now show that, under certain conditions, local densities for the generator in $\M_0$, thought as the vacuum representation, lift to such on $\Mh_\beta$, the thermal representations (including both system and bath).
\begin{theorem}
   Let $\omega_0$, $\omega_\beta$ be two invariant states on the net $\A$, of which $\omega_\beta$ fulfills the $\beta$-KMS property. 
   Suppose that  
   \begin{enumerate}[(i)]
     \item \emph{(Reeh-Schlieder property)} for every nonempty region $\ocal$, the inclusions $\M_0(\ocal) \Omega_0\subset\Hil_0$ and $\M_\beta(\ocal)\Omega_ \beta \subset \Hil_\beta$ are dense;
     \item \emph{(local normality)} for every precompact region $\ocal$ with nonempty $\ocal'$, 
     the representations $\pi_0\restriction\A(\ocal)$ and $\pi_\beta\restriction \A(\ocal)$ are quasi-equivalent,
     \item \emph{(split property)} for every precompact region $\ocal_1$ with nonempty $\ocal_1'$, there exists another precompact region $\mathcal{O}_2 \supset \mathcal{O}_1$ with nonempty $\ocal_2'$  such that the inclusion $\M_0(\ocal_1')'\subset\M_0(\ocal_2')'$ is split.
   \end{enumerate}
If $\M_0$ has local densities $\ell_{0,\ocal}$ in the weak sense, then $\Mh_\beta$ has local densities $\ell_{\beta,\ocal}$ in the weak sense that are split; it holds for $f \in \ccs(\ocal_1)$ that
   \begin{equation*}
        \ell_{\beta,\ocal}(f) = \overline{ V \ell_{0,\ocal}(f) V^\ast - J_\beta V \ell_{0,\ocal}(\bar f) V^\ast J_\beta}
   \end{equation*}
   with a certain unitary $V:\Hil_0\to\Hil_\beta$.
\end{theorem}

\begin{proof}
  Fix $\ocal$, $\ocal_1$ and $\ell_{0,\ocal}$ for the local density in $\M_0(\ocal)$, as in Def.~\ref{def:densityweak}; given these, fix $\ocal_2$ as in property (iii). Let $V \coloneqq V_{\ocal_2}:\Hil_0\to\Hil_\beta$ be the unitary arising from quasi-equivalence on $\A(\ocal_2)$, see property (ii) combined with (i). Now define
  \begin{equation}\label{eq:hfrom0}
      h(f) :=  V \ell_{0,\ocal}(f) V^\ast.
  \end{equation}
  This is, by the remark after \eqref{eq:fieldtransf}, a local field associated with $\M_\beta \restriction \ocal_1$.
  On the other hand,
  \begin{equation}
      \tilde h(f) :=  J_\beta V \ell_{0,\ocal}(\bar f) V^\ast J_\beta
  \end{equation}
  is a local field associated with $\Mt_\beta \restriction \ocal_1$. 
  
  We first show that the difference $h(f)-\tilde h(f)$ is densely defined and closable: With $Y:\Hil_0\to\mathcal{K}\otimes\mathcal{K}$ the unitary implementing the split inclusion, $k(f):=Y V^\ast h(f) V Y^\ast = Y\ell_{0,\ocal}(f)Y^\ast$ is affiliated with  $Y \M_0(\ocal_1')'Y^\ast \subset \mathfrak{B}(\mathcal{K}) \otimes \idop$; hence its domain is of the form $\mathcal{D}_1\otimes\mathcal{K}$. Similarly,  $\tilde k(f):=Y V^\ast \tilde h(f) V Y^\ast$ is affiliated with  $Y V^\ast\Mt_\beta(\spt)V Y^\ast \subset Y \M_0(\ocal_2')Y^\ast  \subset \idop \otimes \mathfrak{B}(\mathcal{K}) $, hence its domain is of the form $\mathcal{K} \otimes \mathcal{D}_2$. Therefore, the sum $k(f) - \tilde k(f)$ is well-defined on the intersection $\dom{k(f)}\cap \dom{\tilde k(f)} =\mathcal{D}_1\otimes\mathcal{D}_2$ and closable there \cite[Sec.~VIII.10]{Reed:1972}. 
  %At the same time, $\mathcal{D}_1\otimes\mathcal{D}_2$ is clearly a core for $k(f)$ and $\tilde k(f)$. 
  By unitary transformation,  $\ell_{\beta,\ocal}(f) \coloneqq \overline{h(f)-\tilde h(f)}$ is well-defined.

  By a similar argument involving a common core that ``factorizes'' along the tensor product, we can also show that $\ell_{\beta,\ocal}(f)$ is linear in $f$.
  
  Evidently, $\ell_{\beta,\ocal}(f)$ is split.  Since $h(f)$ is affiliated with $\M_\beta(\ocal_1')'\cap \M_\beta(\spt)$ and $\tilde h(f)$ with $\Mt_\beta(\ocal_1')'\cap \Mt_\beta(\spt)$, Lemma~\ref{lem:sumaffil} shows that $\ell_{\beta,\ocal}(f)$ is affiliated with the algebra generated by their union, which works out to be $\Mh_\beta(\ocal_1')'$, cf.~Eq.~\eqref{eq:affilunion}.

  Now let $A,B \in \M_\beta(\ocal)\cap\dom{\delta_\beta}$, so that  $\tilde B \coloneqq J_\beta B J_\beta$ is a generic element of $ \Mt_\beta(\ocal)\cap\dom{\delta_\beta}$. From Lemma~\ref{lemma:gentogen}, we have
   \begin{equation}\label{eq:dba}
   \begin{aligned}
      \delta_\beta(A\tilde B)  
      &
      = \delta_\beta(A)\tilde B + A \delta_\beta(\tilde B)
      \\
      &
      =
      V \delta_0( V^\ast A V) V^\ast \tilde B  
     +  A J_\beta V \delta_0(  V^\ast B V) V^\ast J_\beta .
     \end{aligned}
     \end{equation}
     Picking $f$ to be spacelike-one for $\ocal$, the derivation $\delta_0$ in the above equation can be replaced (weakly) by the commutator with $\ell_{0,\ocal}(f)$, or using \eqref{eq:hfrom0}, with $h(f)$. More precisely, let $\vpsi\in\dom{h(f)} \cap \dom{\tilde h(\bar f)}$. 
     Since $h(f)$ is affiliated with $\M_\beta(\ocal_1')'\cap\M_\beta(\spt)$, we know that $\tilde B \in \Mt_\beta(\ocal)$ leaves its domain invariant, hence $\tilde B \vpsi \in \dom{h(f)}$. Likewise, since also $\vpsi\in\dom{\tilde h(f)}$, we find $J_\beta A^\ast \vpsi \in \dom{h(f)}$. One then computes from \eqref{eq:dba}, 
     \begin{equation}
     \begin{aligned}
     (\vpsi, \delta_\beta(A\tilde B) \vpsi)
      &=  i (\vpsi, [ h(f),  A ] \tilde B \vpsi )
      +\overline{i ( J_\beta A^\ast \vpsi, [ h(f),  B  ]  J_\beta \vpsi )}
      \\ &=  i (\vpsi, [ h(f), A ] \tilde B \vpsi )
      -  i (\vpsi, A [ \tilde h(f), \tilde B ]  \vpsi ).
   \end{aligned}
   \end{equation}
   Since $A$ spectrally commutes with $\tilde h(f)$, and $\tilde B$ with $h(f)$, we finally obtain
   \begin{equation}
   (\vpsi, \delta_\beta(A\tilde B) \vpsi)
   = i (\vpsi, [ h(f) - \tilde h(f), A\tilde B ]  \vpsi ), \quad
     \vpsi \in \dom{h(f)} \cap \dom{\tilde h(f)}.
   \end{equation}
Since these $\vpsi$ form a core for $\ell_{\beta,\ocal}(f)$,  we have established 
\begin{equation}\label{eq:psibetac}
   (\vpsi, \delta_\beta(C) \vpsi)
 =  i (\vpsi, [ \ell_{\beta,\ocal}(f), C ] \vpsi )
 \quad \text{for all } \vpsi\in\dom{\ell_{\beta,\ocal}(f)}
\end{equation}
whenever $C=A \tilde B$ of the form above. It then holds also for $C=(A\tilde B)(g)\in\Mh(\ocal)$, i.e., with a small smearing as in \eqref{eq:Bint}. By $\sigma$-weak approximation with smooth elements, we can drop the requirement that $A,B \in \dom \delta_\beta$ (cf.~the remark after Eq.~\eqref{eq:Bint}), i.e., the identity \eqref{eq:psibetac} can be extended to all $C=(A\tilde B)(g)$, $A \in \M_\beta(\ocal)$, $B \in \Mt_\beta(\ocal)$. Since linear combinations of these $A \tilde B$ are weakly dense in $\Mh_\beta(\ocal)$, one then obtains \eqref{eq:psibetac} for all $C=C_0(g)$, $C_0\in\Mh_\beta(\ocal)$. Removing the smearing function for $C_0\in\dom{\delta_\beta}$, one thus has \eqref{eq:psibetac} for all $C \in \dom{\delta_\beta} \cap \Mh_\beta(\ocal)$, as desired. 
\end{proof}

\section{Thermal states of the real scalar free field}\label{sec:freescalar}

As a specific example, we consider the real scalar free field on Minkowski space in a thermal equilibrium state (see, e.g., \cite{Kay:purification,UMT:thermo}). That is, throughout this section, we will set $\spt=\R^4$, $\csf=\{0\} \times \R^3$ (the time-0 surface), and $\Xi$  the usual time evolution. The algebras $\A(\ocal)$ will be the Weyl algebras of the real scalar free field, generated by Weyl operators $W(f)$,  $f=\bar f \in \ccs(\R^4)$. On the quasilocal algebra $\A$, consider the quasifree state given by
\begin{equation}
    \omega(W(f)) = e^{-\omega_2^s(f,f)/2}
\end{equation}
where $\omega_2^s$, the symmetric part of the two point function, has the integral kernel
\begin{equation}
   \omega_2^s(x,y)= \int \frac{\di^3\textbf{k}}{(2\pi)^3}\frac{1}{2\en{k}}\frac{e^{\beta\en{k}}+1}{e^{\beta\en{k}}-1}\cos(k\cdot(x-y)).
\end{equation}
Here we used the notation $\en{k}=\sqrt{\|\mathbf{k}\|^2+m^2}$ and $k=(\en{k},\textbf{k})$.
The automorphism $\alpha_t$ acts by translating the test functions $f$ by $t$ in $x^0$-direction, and the state $\omega$ is KMS at inverse temperature $\beta$. 

We can now study the GNS representation of $\A$ induced by the state $\omega$. This yields a Fock representation on a tensor product $\mathcal{F}^2$ of two copies of the usual symmetrised bosonic Fock space, with two pairs $\cre{a}{k},\ani{a}{k}$,  $\cre{b}{k},\ani{b}{k}$ of creators and annihilators  acting on it, and with vacuum vector $\Omega=\Omega_b\otimes\Omega_a$ (see \autoref{sec:edaffil} for more details on the notation). In this representation, $L$ is given by
\begin{equation}\label{eq:L}
    L=\int \frac{\di^3\mathbf{k}}{(2\pi)^3}\en{k}(\cre{b}{k}\ani{b}{k}-\cre{a}{k}\ani{a}{k})
\end{equation}  
as a sesquilinear form on the domain $\domSchwartz \subset \mathcal{F}^2$ of vectors with finite particle number, all components of which are Schwartz functions, cf.~Eq.~\eqref{domain:DS} below.
$J$ acts as
\begin{equation}\label{eq:jact}
    J( \vpsi\otimes\vxi)=\Gamma(C)\vxi\otimes\Gamma(C)\vpsi,\;\;\;\vpsi\otimes\vxi\in\mathcal{F}^2,
\end{equation}
where $\Gamma(C)$ is the second quantization of the complex conjugation operator.

By differentiation\footnote{Note that the representation is regular since it is induced by a quasifree state.}, we can obtain a quantum field $\phi$ with formal kernel
\begin{equation}\label{eq:fielddef}
\begin{aligned}
    \phi(x)&=\int\frac{\di^3\mathbf{k}}{(2\pi)^3}\frac{1}{\sqrt{2\en{k}}}\left[\bsh{-}{k}\ani{a}{k}e^{ikx}+\bsh{-}{k}\cre{a}{k}e^{-ikx}+\bsh{+}{k}\ani{b}{k}e^{-ikx}+\bsh{+}{k}\cre{b}{k}e^{ikx}\right]
\end{aligned}
\end{equation}
where $\bsh{\pm}{k}=(\pm 1 \mp e^{\mp\beta\en{k}} )^{-1/2}$. When smeared with test functions and closure, the $\phi(f)$ are then affiliated with $ \M(\ocal')'\cap\M(\R^4)$ whenever $\operatorname{supp} f \subset \ocal$, hence they are quantum fields in our sense. Analogous to the vacuum case \cite[Sec.~X.7]{ReedSimon:1975-2}, $\domSchwartz$ is a common invariant core for all $\phi(f)$, consisting of analytic vectors (see \autoref{sec:edaffil} below). From there, it can be shown that elements of the polynomial algebra $\polyalg(\ocal)$, i.e., sums and products of the $\phi(f)$, $\operatorname{supp} f \subset \ocal$, are affiliated with $\M(\ocal')'\cap \M(\R^4)$: Knowing that $P\in\polyalg(\ocal)$ commutes with $\phi(g)$, $\operatorname{supp} g \subset \ocal'$ and with all $J\phi(g)J$ strongly on $\domSchwartz$, the argument works similar to Proposition~\ref{prop: final_proof_affil} below.

We now turn to the energy density. We consider the following ``normal ordered product'' which is formally analogous to the energy density in the vacuum,
\begin{equation}\label{eq:hdef}
    h(x)=\frac{1}{2}:(\partial_0\phi(x))^2:+\frac{1}{2}\sum_{i=1}^3:(\partial_i\phi(x))^2:+\frac{1}{2}m^2:(\phi(x))^2:.
\end{equation}
We can show that, after integration with test functions, this is indeed a covariant symmetric  quantum field associated with the net $\M$, in the sense of Def.~\ref{def:associated}.   The somewhat technical proof will be deferred to  \autoref{sec:edaffil}; here we just remark that $\domSchwartz$ is a common core for all $h(f)$, $f \in \testf$.

We now set 
\begin{equation}\label{eq:ellexpl}
    \ell(x) \coloneqq h(x) - J h(x) J,
\end{equation}
first as sesquilinear forms on the domain $\domSchwartz$; thanks to $J$-invariance of  $\domSchwartz$, this yields (after operator closure, see Proposition~\ref{prop:ellsum}) a symmetric, covariant and split quantum field $\ell$ associated with $\Mh$, with $\domSplit=\domSchwartz$.
We show that $\ell$ is actually a Liouvillian density for the generator $L$ given in \eqref{eq:L}; by contrast, the integral of $h$ over all space does not yield a reasonable operator.

\begin{proposition}
   $\ell$ is a density for $L$ in the strong sense  (Def.~\ref{def:integration}, with $\mathcal{D}=\domSchwartz$).
\end{proposition}

\begin{proof} 
$\domSchwartz$ is a common core for all $\ell(f)$ and for $L$, invariant under $U(t)$. Property (\ref{it:fieldcont}) of Def.~\ref{def:integration} is checked by dominated convergence, while property (\ref{it:fieldint}) is clear by definition of $\ell$. It remains to show property (\ref{it:ellint}), i.e., that for all $\vpsi\in\domSchwartz$,
\begin{equation}
(\vpsi, L \vpsi) = \int \big(\vpsi, \ell(0,\mathbf{x}) \vpsi\big) \,\di^3\mathbf{x} .
\end{equation}
This follows from a computation which we summarize in \autoref{sec:integral} below.
\end{proof}

We can now show an $L^4$ inequality for $\ell$.

\begin{theorem}
 Let $f = g^2$ with some $g=\bar g \in \ccs(\R^4) \backslash \{0\}$. Then $\ell(f)$ as in \eqref{eq:ellexpl} fulfills a nontrivial quantum $L^4(\M(\R^4),\Omega)$ inequality.
\end{theorem}
\begin{proof}
  We first note that $h(f)$, as a sesquilinear form on $\domP$, is $L^4$-bounded below. This can be extracted from the literature: As a quasifree KMS state, $\omega=(\Omega,\,\cdot\, \Omega)$ is Hadamard \cite{SV:passivity}, hence so is any state induced by a vector in $\domP$ \cite[Theorem~4.5]{BFK:spectrum}. Now in quasifree Hadamard states, $h(g^2)$ -- as a  ``sum of Wick squares'' -- fulfills a state-independent quantum energy inequality \cite[Sec.~3]{Few:qft_inequalities}, or in our language, a lower $L^2$ bound; a forteriori, we find a lower $L^4$ bound.

   On the other hand, $h(f)$ is not $L^4$-bounded above on the same domain. To that end, we exhibit an explicit sequence of vectors in $\domP$ which is violating these bounds; see \autoref{sec:nontriv} below, in particular Proposition~\ref{prop:nontrivl4}.
   In summary, $h(f)$ fulfills a nontrivial quantum $L^4$ inequality. 
   
   Now consider the split field $\ell$, where we can take $\domSplit \coloneqq\domP\subset\domSchwartz$ (which is inside the domains of $h(f)$ and of $Jh(f)J$). We note that $\domP\subset \domAffil{J h(f) J}$ since all operators in $\polyalg(\R^4)$ are affiliated with $\M(\R^4)$ and have $\domP$ as a common invariant dense domain. We can therefore apply Theorem~\ref{thm:splitdensity}, showing that $\ell(f)$, as sesquilinear form on $\mathcal{D}=\domP$, fulfills a nontrivial $L^4$ inequality.
\end{proof}

We remark that $\ell(g^2)$ is, in general, not $L^2$-bounded below, hence it does not fulfill a quantum energy inequality in the usual sense. This is again verified by means of an explicit counterexample, which will be given in \autoref{sec:nontriv} below (see Proposition~\ref{prop:nontrivL2}).

\subsection{Affiliation of the Hamiltonian density}\label{sec:edaffil}

In this section, we show that the ``Hamiltonian density'' $h(x)$, as defined in \eqref{eq:hdef}, yields a well-defined quantum field in our sense. The main technical point is affiliation of the ``smeared'' field $h(f)$ to the relevant Weyl algebras, for which detailed estimates of its matrix elements are required. We generalize known methods from the vacuum setting \cite{LangerholcSchroer} to our thermal case. 

We first fix our notation: Let $\mathcal{F}^2$ be the ``doubled Fock space'', \begin{equation}
    \mathcal{F}^2\coloneqq \mathcal{F}^s( L^2(\R^3))\bigotimes\mathcal{F}^s( L^2(\R^3)),
\end{equation}
where $\mathcal{F}^s$ denotes the symmetric Fock space. Its elements can be thought of as double sequences $\vpsi = (\vpsi^{(n_p,n_h)})_{n_p,n_h\geq 0}$, where $\vpsi^{(n_p,n_h)} \in L^2(\R^{3n_p+ 3n_h})$, symmetric in the first $3n_p$ and in the last $3n_h$ arguments; $n_p$ is interpreted as the ``number of particles'' and $n_h$ the ``number of holes'' over the thermal background state, $\Omega=(1,0,0,\ldots)$. We will often consider the dense subspace 
\begin{equation}\label{domain:DS}
\begin{aligned}
        \domSchwartz\coloneqq \big\{ \vpsi\in \mathcal{F}^2 \;\big\vert\; &\vpsi^{(n_p,n_h)}\in\mathcal{S}(\R^{3n_p+ 3n_h})\;\text{for all } n_p,n_h \in \mathbb{N}_0; 
        \\
        &\vpsi^{(n_p,n_h)}=0\;\text{for almost all } n_p,n_h  \in \mathbb{N}_0 \big\}.
\end{aligned}
\end{equation}
Following \cite{ReedSimon:1975-2}, the annihilators $\ani{a}{k}$ and $\ani{b}{k}$ are well-defined operators on the domain $\domSchwartz$:
\begin{equation}\label{eq:annihidef}
\begin{aligned}
    (\ani{a}{k}\vpsi)^{(n_p,n_h)}(\textbf{s}_1,\ldots,\textbf{s}_{n_p};\textbf{t}_1,\ldots,\textbf{t}_{n_h})&=(2\pi)^{3/2}\sqrt{n_h+1}\vpsi^{(n_p,n_h+1)}(\textbf{s}_1,\ldots,\textbf{s}_{n_p};k,\textbf{t}_1,\ldots,\textbf{t}_{n_h}),\\
    (\ani{b}{k}\vpsi)^{(n_p,n_h)}(\textbf{s}_1,\ldots,\textbf{s}_{n_p};\textbf{t}_1,\ldots,\textbf{t}_{n_h})&=(2\pi)^{3/2}\sqrt{n_p+1}\vpsi^{(n_p+1,n_h)}(k,\textbf{s}_1,\ldots,\textbf{s}_{n_p};\textbf{t}_1,\ldots,\textbf{t}_{n_h}),
\end{aligned}
\end{equation}
while their adjoints, the creators, are defined only as sesquilinear forms on $\domSchwartz\times \domSchwartz$, given formally by
\begin{equation}\label{eq:creatordef}
\begin{aligned}
    (\cre{a}{k}\vpsi)^{(n_p,n_h)}(\textbf{s}_1,\ldots,\textbf{s}_{n_p};\textbf{t}_1,\ldots,\textbf{t}_{n_h})=&\\
    (2\pi)^{3/2}\frac{1}{\sqrt{n_h}}\sum_{l=1}^{n_h}\delta^3(\textbf{k}-\textbf{t}_l)&\vpsi^{(n_p,n_h-1)}(\textbf{s}_1,\ldots,\textbf{s}_{n_p};\textbf{t}_1,\ldots,\textbf{t}_{l-1},\textbf{t}_{l+1},\ldots ,\textbf{t}_{n_h});\\
    (\cre{b}{k}\vpsi)^{(n_p,n_h)}(\textbf{s}_1,\ldots,\textbf{s}_{n_p};\textbf{t}_1,\ldots,\textbf{t}_{n_h})=& \\
    (2\pi)^{3/2}\frac{1}{\sqrt{n_p}}\sum_{l=1}^{n_p}\delta^3(\textbf{k}-\textbf{s}_l)&\vpsi^{(n_p-1,n_h)}(\textbf{s}_1,\ldots,\textbf{s}_{l-1},\textbf{s}_{l+1},\dots, \textbf{s}_{n_p};,\textbf{t}_1,\ldots,\textbf{t}_{n_h}).
\end{aligned}
\end{equation}

Before turning to $h(f)$, we define the following auxiliary quantities. Let $f \in \mathcal{S}(\R^4)$ be real-valued, and let $\hat f$ be its 4-dimensional Fourier transform. Let $\polyn(p,k)$ be a polynomial in the variables $p,k\in\R^4$, and let  $\mu,\nu,\rho,\sigma \in \{0,1\}$ with the condition $\nu\geq \sigma$. We define the sesquilinear form $O_\polyn^{\mu\nu\rho\sigma}(f)$ on the dense domain $\domSchwartz \times \domSchwartz$ by
\begin{equation}\label{Eq: def_via_quadratic}
\begin{aligned}
       O_\polyn^{\mu\nu\rho\sigma}(f)=
       \iint\frac{\di^3\textbf{k}\di^3\textbf{p}}{2\sqrt{\en{p}\en{k}}} & \frac{\polyn(p,k)}{(2\pi)^6}\hat{f}\left((-1)^{\rho+\sigma+1}\en{p}+(-1)^{\mu+\nu+1}\en{k},\textbf{p}+\textbf{k}\right)\times\\
    &\times(\mathscr{B}_{\textbf{p}}^+)^{\rho}(\mathscr{B}_{\textbf{p}}^-)^{1-
    \rho}(\mathscr{B}_{\textbf{k}}^+)^{\mu}(\mathscr{B}_{\textbf{k}}^-)^{1-\mu}
   c^{\mu\nu}_\textbf{k} c^{\rho\sigma}_\textbf{p},
    \end{aligned}
    \end{equation}
where  
\begin{equation}
c_{\textbf{k}}^{00}=a_{\textbf{k}},\;\;\;c_{\textbf{k}}^{10}=b_{\textbf{k}}\;\;\;c_{\textbf{k}}^{01}=a^\dagger_{\textbf{k}}\;\;\;c_{\textbf{k}}^{11}=b^\dagger_{\textbf{k}}.
\end{equation}
These sesquilinear forms can actually be understood as unbounded operators, as we now show.
\begin{lemma}\label{lem: first_estimate}
    For every real-valued $f\in\ccs(\R^4)$, the form $O_\polyn^{\mu\nu\rho\sigma}(f)$ extends to an operator on $\domSchwartz$, and 
    for $\vpsi = (0,\ldots,0,\vpsi^{(n_p,n_h)},0,\ldots)\in\domSchwartz$, the following estimate holds:
    \begin{equation*}
       \|O_\polyn^{\mu\nu\rho\sigma}(f)\vpsi \|\leq C^2_{f,\beta,\polyn} \sqrt{n_\#+2}\sqrt{n_\#+1}\big\|\vpsi^{(n_p,n_h)}\big\|'^s_{2},
    \end{equation*}
     where  $n_\#=\operatorname{max}\{n_p,n_h\}$, $s=\deg \polyn$, and $\|{\cdot}\|'^s_{2}$ is the $L_2$ norm with respect to the measure $\prod_{i=1}^n \di^3 \mathbf{p}_i(\|\mathbf{p}_i\|^2+1)^{s+4}$ on $\R^{3n_p+3n_h}$. Further, $C^2_{f,\beta,\polyn}$ denotes a positive constant that depends only on the function $f$, the polynomial $\polyn$ and the parameter $\beta$. 
\end{lemma}
\begin{proof}
    We will treat three separate cases, corresponding to the presence of two creators ($\nu=\sigma=1$), two annihilators ($\nu=\sigma=0$) or one creator and one annihilator ($\nu=1,\sigma=0$).
    
    We start with the case $\nu=\sigma=1$, obtaining a result that generalises the one in \cite[Appendix 1]{LangerholcSchroer} for the vacuum representation. Inserting \eqref{eq:creatordef} and using Cauchy-Schwarz, we find for $\vphi\in\domSchwartz$: 
    \begin{equation}\label{eq: norm_2_creation}
         \left\lvert (\vphi, O_\polyn^{\mu 1 \rho 1}(f)\vpsi ) \right\rvert \leq  \|\vphi^{(n_p+\rho+\mu,n_h+2-\rho-\mu)} \|_2 \sqrt{n_\#+2} \sqrt{n_\#+1}\big\|\hat{f'} \big\|_{2}\big\| \vpsi^{(n_p,n_h)}\big\|_{2},
    \end{equation}
    where we denoted
    \begin{equation}\label{eq:fprimedef}
    \hat{f'}(\textbf{p},\textbf{k})\coloneqq \frac{\mathcal{P}(\en{p},\mathbf{p},\en{k},\mathbf{k})}{(2\pi)^3}\frac{\hat{f}\left((-1)^{\rho}\en{p}+(-1)^{\mu}\en{k},\textbf{p}+\textbf{k}\right)}{2\sqrt{\en{p}\en{k}}}
    (\mathscr{B}_{\textbf{p}}^{+})^{\rho}
    (\mathscr{B}_{\textbf{p}}^{-})^{1-\rho}
    (\mathscr{B}_{\textbf{k}}^{+})^{\mu}
    (\mathscr{B}_{\textbf{k}}^{-})^{1-\mu}.
\end{equation}
By the Riesz representation theorem,  $O_\polyn^{\mu 1 \rho 1}(f)$ thus extends to an operator with the stated bound \emph{if} we can show that $\hat{f'}$ is indeed square-integrable. To that end, note that for $\rho=\mu$, the first argument of $\hat f$ in \eqref{eq:fprimedef} is growing linearly with $\|\mathbf{p}\|$ and $\|\mathbf{k}\|$, hence $\hat f(\ldots)$ is decaying faster than polynomially, yielding the result. If $\rho \neq \mu$, then at least one of the factor $\mathscr{B}$ decays exponentially, while $\hat f$ decays faster than polynomially in its first argument, $\pm(\en{p}-\en{k})$; this is enough to show that $\hat f'$ is square integrable (cf. \cite[Appendix~B]{Sangaletti:thesis} for details). 

For the case $\nu=\mu=0$, we note that (as form adjoints)
\begin{equation}
    (O_\polyn^{\mu 1\rho 1}(f))^\dagger= O_{\overline\polyn}^{\mu 0\rho 0}(\overline{f_R})
\end{equation}
where $f_R(t,\mathbf{x})=f_R(t,-\mathbf{x})$.
Applying the above arguments again (but with the roles of $\vphi$ and $\vpsi$ reversed) then yields the proposed result.

The remaining case is $\nu=1,\sigma=0$. Using \eqref{eq:annihidef}, \eqref{eq:creatordef} and the Cauchy-Schwarz inequality, we obtain
\begin{equation}\label{eq: norm_cre_an}
   \lvert (\vphi, O_\polyn^{\mu 1\rho 0}(f)  \vpsi) \rvert \leq 
    \big\| \vphi^{(n_p+\mu-\rho,n_h-\mu+\rho)}\big\|_{2}
    \sqrt{n_\#} \sqrt{n_\#}\left\|{\hat{f''}}\right\|_{2}\big\| \vpsi^{(n_p,n_h)}\big\|_{2}^{\prime s},
\end{equation} 
where 
\begin{equation}
    \hat{f''}(\textbf{p},\textbf{k})\coloneqq \frac{\mathcal{P}(\en{p},\mathbf{p},\en{k},\mathbf{k})}{(2\pi)^3}\frac{\hat{f}\left((-1)^{\rho+1}\en{p}+(-1)^{\mu}\en{k},\textbf{p}+\textbf{k}\right)}{2\sqrt{\en{p}\en{k}}(\|\textbf{p}\|^2+1)^{s+4}}
    (\mathscr{B}_{\textbf{p}}^{+})^{\rho}
    (\mathscr{B}_{\textbf{p}}^{-})^{1-\rho}
    (\mathscr{B}_{\textbf{k}}^{+})^{\mu}
    (\mathscr{B}_{\textbf{k}}^{-})^{1-\mu},
\end{equation}
and it remains to show that $f''\in L_2(\R^6)$. In fact, for $\mu \neq\rho$ and for $\mu=\rho=0$, this follows with arguments as above; while for $\mu=\rho=1$, the additional factor $(\|\textbf{p}\|^2+1)^{s+4}$ in the denominator, together with decay of $\hat f(\ldots)$ in the difference variable, is enough to guarantee finite $L^2$ norm.
\end{proof}

We now turn to the sesquilinear form $h(f)=\int f(x) h(x) dx$ with $h(x)$ as in \eqref{eq:hdef}. As a consequence of the above estimates, we can extend $h(f)$ to an operator:

\begin{lemma}\label{Lem: estimate:energy_density}
    For every real valued $f\in\test$, the sesquilinear form $h(f)$ extends to a symmetric operator on $\domSchwartz$; and for every $\vpsi = (0,\ldots,0,\vpsi^{(n_p,n_h)},0,\ldots)\in\domSchwartz$, the following estimate holds:
    \begin{equation}\label{eq: estimate_energy_density}
         \| h(f)\vpsi\| \leq B^2_{f,\beta} \sqrt{n_\#+2}\sqrt{n_\#+1}\left\|\vpsi^{(n_p,n_h)}\right\|'_{2,\R^{3n_p+3n_h}},
    \end{equation}
    where $n_\#=\operatorname{max}\{n_p,n_h\}$, where $\|\cdot\|'_{2}$ is the $L_2$ norm with respect to the measure $\prod_{i=1}^{n_p+n_h} \di^3 \mathbf{p}_i(\|\mathbf{p}_i\|^2+1)^{6}$ on $\R^{3n_p+3n_h}$ and $B_{f,\beta}^2$ is a positive constant.
\end{lemma}
\begin{proof}
     $ h(f)$ can be written as a linear combination of operators $O^{\mu\nu\rho\sigma}_\polyn(f)$ with $\deg \polyn \in \{0,2\}$. Using Lemma~\ref{lem: first_estimate} together with the triangle inequality, the estimate~\eqref{eq: estimate_energy_density} is immediately derived. In addition, the form $h(x)$ in \eqref{eq:hdef} is hermitean, and therefore the operator $h(f)$ is symmetric on $\domSchwartz$ when $f$ is real-valued. 
\end{proof}
As an immediate consequence of the previous lemma, the operator $h(f)$ on $\domSchwartz$ is closable; we will denote its closure with the same symbol $h(f)$. 

As a next step, we show that the Weyl operators $\pi(W(g))=\exp i \phi(g)$ and $\tilde\pi(W(g))$ map the core $\domSchwartz$ into $\dom {h(f)}$. Generalising known results (see for instance \cite{ReedSimon:1975-2,BraRob:qsm2}) to the thermal representation, it is easy to see that the following estimate holds for every real-valued function $g\in\ccs(\mathbb{R}^4)$:
\begin{equation}\label{eq: est_for_field}
    \|\phi(g) \vpsi \|\leq \frac{2^2}{(2\pi)^{3}}\bsh{+}{0}\sqrt{n^\#+1}\Big\|\frac{\hat{g}}{\sqrt{2\omega}}\restriction{\mathcal{H}_m}\Big\|_2\|\vpsi^{(n_p,n_h)}\|,
\end{equation}
where $\mathcal{H}_m$ denotes the positive mass hyperboloid and we made use of the relation
\begin{equation}
    \bshs{-}{p}\leq\bshs{+}{p}\leq \bshs{+}{0} = \frac{1}{1-e^{-\beta m}}.
\end{equation}
Further, since $g\in\test$, its Fourier transform restricted to the mass hyperboloid belongs to $\mathcal{S}(\R^3)$, which implies  $\phi(g) \vpsi \in\domSchwartz$.
Iterating this $n$ times, we obtain:
\begin{equation}\label{eq: estimate_for_an}
     \|(\phi (g))^n \vpsi\|\leq \frac{2^{(2n)}}{(2\pi)^{3n}}(\bsh{+}{0})^{n}\sqrt{n^\#+1}\ldots\sqrt{n^\#+n}\Big\|\frac{\hat{g}}{\sqrt{2\omega}}\restriction{\mathcal{H}_m}\Big\|_2^{n}\|\vpsi^{(n_p,n_h)}\|.
\end{equation}
This estimate shows that $\domSchwartz$ is a dense set of analytic vectors for the field operators $\phi(g)$; namely, by the ratio test, the following series is convergent in the Hilbert space norm,
\begin{equation}
    e^{i\phi(g)}\vpsi =\sum_{n=0}^{\infty}\frac{(i\phi(g))^n}{n!}\vpsi,
\end{equation}
and can be used to expand the Weyl operator acting on a vector $\vpsi\in\domSchwartz$.
(By the same argument, all finite particle number vectors are analytic for $\phi(g)$.)
We can now prove the following lemma:
\begin{lemma}\label{lem: dense_in_dom}
    For any $f \in \test$, and any real-valued $g\in\test$, the Weyl operators $\pi(W(g))$ and $\tilde\pi(W(g))$ map $\domSchwartz$ into $\dom{h(f)}$. 
\end{lemma}
\begin{proof}
    Let $\vpsi\in\domSchwartz$; by linearity, it suffices to consider $\vpsi = (0,\ldots,0,\vpsi^{(n_p,n_h)},0,\ldots)$. 
    For each finite $N$ the vector
    \begin{equation}\label{eq: conv_seq}
        \vpsi^N\coloneqq \sum_{n=0}^{N}\frac{(i\phi(g))^n}{n!}\vpsi
    \end{equation}
    lies in $\domSchwartz$, hence in the domain of $h(f)$.
    Thanks to the estimate \eqref{eq: estimate_energy_density}, we have:
    \begin{equation}\begin{aligned}
       &\left\| h(f) \sum_{n=0}^{N}\frac{(i\phi(g))^n}{n!}\vpsi \right\|\leq \sum_{n=0}^N\frac{\| h(f) (\phi(g))^n \vpsi\|}{n!}\\
       &\leq \sum_{n=0}^N\frac{B^2_{f,\beta}}{n!}\frac{2^{(2n)}}{(2\pi)^{3n}}(\bsh{+}{0})^{n}\sqrt{n^\#+1}\ldots\sqrt{n^\#+n+2}\Big\|\frac{\hat{g}}{\sqrt{2\omega}}\restriction{\mathcal{H}_m}\Big\|^{\prime n}\|\vpsi^{(n_p,n_h)}\|'.
    \end{aligned}\end{equation}
    By the ratio test, this sequence converges as $N\to\infty$. In conclusion, within the closed graph of $h(f)$:
    \begin{equation}\label{eq: conv_in_graph_anali}
        \vpsi^N\oplus  h(f)\vpsi^N\xrightarrow[]{N\to\infty}\pi(W(g)) \vpsi \oplus h(f) \pi(W(g))\vpsi.
    \end{equation}
    We conclude that $\pi(W(g))\vpsi\in\dom{h(f)}$.---The argument for $\tilde\pi(W(g)) = J \pi(W(g)) J$ is analogous, considering that $J$ acts as in \eqref{eq:jact}.
\end{proof}
We can now prove the affiliation of the energy density operator:
\begin{proposition}\label{prop: final_proof_affil}
    Let $f\in\ccs(\ocal)$. Then $h(f)$ is affiliated with $\M(\ocal')' \cap \M$.
\end{proposition}
\begin{proof}
    Let $\vpsi\in \domSchwartz$. One checks by direct computation\footnote{The computation is first done as sesquilinear forms on $\domSchwartz \times \domSchwartz$; perhaps the easiest approach is to understand the normal ordering in \eqref{eq:hdef} in a point-splitting/ground state subtraction sense, and use the locality of the field $\phi$.} and Lemma~\ref{Lem: estimate:energy_density} that 
    \begin{equation}
         \left[ h(f),\phi(g)^n\right]\vpsi=0\;\;\; \text{if } \operatorname{supp} g \subset\ocal'.
    \end{equation}
    Since $\vpsi\in\domSchwartz$ is analytic for $\phi(g)$, we have
    \begin{equation}
      \begin{aligned}
        h(f) \pi(W(g)) \vpsi
        &=  h(f) \sum_{n=0}^{\infty}\frac{(i\phi(g))^n}{n!}\vpsi
        =\sum_{n=0}^{\infty} h(f)\frac{(i\phi(g))^n}{n!}\vpsi
        \\
        &=\sum_{n=0}^{\infty}\frac{(i\phi(g))^n}{n!} h(f)\vpsi=\pi(W(g))h(f) \vpsi,
      \end{aligned}
    \end{equation}
    where the second equation is implied by \eqref{eq: conv_in_graph_anali} and the fourth one again employs analyticity of $h(f)\vpsi$ for $\phi(f)$.  Since these $\pi(W(g))$ generate $\M(\ocal')$, Lemma \ref{lem: weaker_affiliation} now shows that $h(f) \affil \M(\ocal')'$.

    Similarly, one finds that $[ h(f), J \phi(g)^n J] = 0$ for any $g\in\test$ as a relation on $\domSchwartz$, and concludes that $h(f) \affil (J \M(\R^4) J)'=\M(\R^4)$. 
\end{proof}

\subsection{Integration to the Liouvillian}\label{sec:integral}

In this section, we show that for the sesquilinear form $\ell(x)$ and the operator $L$ defined at the beginning of \autoref{sec:freescalar}, one has 
\begin{equation}
   \int \di^3\mathbf{x} \,(\vpsi, \ell(0,\mathbf{x}) \vxi) = (\vpsi, L \vxi) \quad \text{for any $\vpsi,\vxi\in\domSchwartz$}.
\end{equation}
This follows with a computation which has been worked out in \cite[Appendix~A]{Sangaletti:thesis}. For the benefit of the reader, we summarize the main steps here.

First, consider the case that $\vpsi$ and $\vxi$ have only one nonzero component, $\vpsi^{(n_p,n_h)}$ and $\vxi^{(n_p+2,n_h)}$ respectively. In that case, one obtains from the expression \eqref{eq:hdef},
\begin{equation}
  (\vpsi, h(0,\mathbf{x})\vxi) = \int\frac{\di^3\mathbf{k} \, \di^3\mathbf{p}}{(2\pi)^6}
  \frac{-\en{k}\en{p} - \mathbf{k}\cdot\mathbf{p}  +m^2}{4\sqrt{\en{k}\en{p}}}
  \bsh{-}{k}\bsh{-}{p}
  (\vpsi, \ani{a}{k}\ani{a}{p} \vxi) e^{-i (\mathbf{k} + \mathbf{p}) \cdot \mathbf{x}}.
\end{equation}
Note here that $(\mathbf{k},\mathbf{p})\mapsto(\vpsi, \ani{a}{k}\ani{a}{p} \vxi)$ is a Schwartz function; hence $\mathbf{x} \mapsto(\vpsi, h(0,\mathbf{x})\vpsi) $ is Schwartz, and integrating it over all $\mathbf{x}\in\R^3$ amounts to evaluating its Fourier transform, i.e., the $(\mathbf{k},\mathbf{p})$-integrand, at $\mathbf{k}+\mathbf{p}=0$, where however the first factor of the integrand vanishes.

Thus in  $\int \di^3\mathbf{x} \,(\vpsi,  h(0,\mathbf{x}) \vxi)$ for generic $\vpsi,\vxi$, the contribution of $\ani{a}{k}\ani{a}{p}$ vanishes; and similarly, the contributions of  $\ani{b}{k}\ani{b}{p}$, $\cre{a}{k}\cre{a}{p}$, $\cre{b}{k}\cre{b}{p}$, and analogous terms in $\int \di^3\mathbf{x} \,(\vpsi, J h(0,\mathbf{x}) J \vxi)$.

Now let only $\vpsi^{(n_p,n_h)}$ and $\vxi^{(n_p+1,n_h+1)}$ be nonzero. Then,  inserting $\ell(x)=h(x) - Jh(x)J$,
\begin{equation}
  (\vpsi, \ell(0,\mathbf{x}) \vxi) = \int\frac{\di^3\mathbf{k} \di^3\mathbf{p}}{(2\pi)^6}
   \frac{\en{k}\en{p} + \mathbf{k}\cdot\mathbf{p}  +m^2}{4\sqrt{\en{k}\en{p}}}
 \Big( \bsh{-}{k}\bsh{+}{p} 
  e^{i (\mathbf{p} - \mathbf{k}) \cdot \mathbf{x}}
  -
  \bsh{+}{k}\bsh{-}{p} 
  e^{i (\mathbf{k} - \mathbf{p}) \cdot \mathbf{x}}
  \Big)
  (\vpsi, \ani{a}{k}\ani{b}{p} \vxi).
\end{equation}
Integration over $\mathbf{x}\in\R^3$ means evaluating the integrand at $\mathbf{k}-\mathbf{p}=0$, where the second factor of the integrand vanishes. Thus the contributions to 
$\int \di^3 x (\vpsi, \ell(0,\mathbf{x}) \vxi)$ of $\ani{a}{k}\ani{b}{p}$, and analogously $\cre{a}{k}\cre{b}{p}$, $\cre{a}{k}\ani{b}{p}$, and $\cre{b}{k}\ani{a}{p}$, all cancel out.

Finally, let only $\vpsi^{(n_p,n_h)}$ and $\vxi^{(n_p,n_h)}$ be nonzero. Then 
\begin{equation}
\begin{aligned}
  (\vpsi,  &\ell(0,\mathbf{x}) \vxi) = \int  \frac{\di^3\mathbf{k} \di^3\mathbf{p}}{(2\pi)^6}
  \frac{\en{k}\en{p} + \mathbf{k}\cdot\mathbf{p}  +m^2}{2\sqrt{\en{k}\en{p}}}
  \\
  &\times \Big(\bsh{-}{k}\bsh{-}{p}-\bsh{+}{k}\bsh{+}{p}\Big)
  \Big( (\vpsi, \cre{a}{k}\ani{a}{p} \vxi) e^{i (\mathbf{k} - \mathbf{p}) \cdot \mathbf{x}} - (\vpsi, \cre{b}{k}\ani{b}{p} \vxi) e^{i (\mathbf{p} - \mathbf{k}) \cdot \mathbf{x}} \Big).
\end{aligned}
\end{equation}
Integrating over $\mathbf{x}$, thus evaluating at $\mathbf{k}=\mathbf{p}$ (up to a multiplicative factor $(2\pi)^3$), yields with ${\bshs{+}{k}}-{\bshs{-}{k}}=1$,
\begin{equation}
\int \di^3 x  (\vpsi,  \ell(x) \vxi) = \int  \frac{\di^3\mathbf{k}}{(2\pi)^3}
  \en{k} 
  \Big( (\vpsi, -\cre{a}{k}\ani{a}{k} \vxi) + (\vpsi, \cre{b}{k}\ani{b}{k} \vxi) \Big).
\end{equation}
Taking linear combinations, this then holds for all $\vpsi,\vxi\in\domSchwartz$, as claimed.

\subsection{Violation of bounds}\label{sec:nontriv}

In this section we study the expectation value of the Liouvillian density $\ell(f)$ smeared by a positive test function $f\in\ccs(\R^4)$, and show that it is neither $L^4$-bounded above nor $L^2$-bounded below. To that end, we give explicit examples of state vectors in which the bounds are violated. 

Interestingly, it suffices to consider one-particle/one-hole states, i.e., those of the form $\vpsi=\phi(g)\Omega$, $g \in \mathcal{S}(\R^4)$. For these, one computes
\begin{align}
\label{eq:spnorm}
\lVert \vpsi  \rVert^2 &= 
\int\frac{\di^3\textbf{k}}{(2\pi)^3}
\frac{1}{2\en{k}}
\Big( \bshs{+}{k} \lvert \hat g(k) \rvert^2 + \bshs{-}{k} \lvert \hat g(-k) \rvert^2  \Big),
\\ \label{eq:sph}
(\vpsi, h(x) \vpsi) &= 
\int\frac{\di^3\textbf{k}\,\di^3\textbf{p}}{(2\pi)^6}
\Big(
\epsilon_m(k,p)
\bshs{+}{k}\bshs{+}{p} \overline{\hat g(k)} \hat g(p) e^{i(k-p)x} 
\\ \notag &\quad + 
\epsilon_m(k,p)
\bshs{-}{k}\bshs{-}{p} \overline{\hat g(-k)} \hat g(-p)\  e^{i(p-k)x} 
\\ \notag &\quad +
\epsilon_m(k,-p)
\bshs{-}{k}\bshs{+}{p} \big(\overline{\hat g(-k)} \hat g(p)\  e^{i(p+k)x} + \mathrm{h.c.}\big) \Big) ,
\\\label{eq:spht}
(\vpsi, Jh(x)J \vpsi) &= 
\int\frac{\di^3\textbf{k}\,\di^3\textbf{p}}{(2\pi)^6}
\bsh{+}{k}\bsh{-}{k}\bsh{+}{p}\bsh{-}{p}
\Big(
\epsilon_m(k,p) \overline{\hat g(k)} \hat g(p) e^{i(k-p)x} 
\\ \nonumber 
&\quad 
 +
  \epsilon_m(k,p) \overline{\hat g(-k)} \hat g(-p)\  e^{i(p-k)x} 
  + \epsilon_m(k,-p)  \big(\overline{\hat g(-k)} \hat g(p)\  e^{i(p+k)x} + \mathrm{h.c.}\big)
\Big),
\end{align}
where $\epsilon_m(k,p) \coloneqq (k^0p^0+\mathbf{k}\cdot\mathbf{p} + m^2)/(2\en{p}\en{k})$ and, as usual, $k=(\en{k},\mathbf{k})$.

Using sequences of such vectors, we can now show that \emph{upper} $L^4$ bounds on $h(f)$ do not exist.

\begin{proposition}\label{prop:nontrivl4}
      Let $f \in \test$ be nonnegative and not identically zero. Then the sesquilinear form $h(f)$ on the domain $\domP  \times \domP$ is not $L^4$-bounded above.
\end{proposition}

\begin{proof}
By translation covariance, we may assume $f(0)>0$. We pick a real-valued $\chi\in\mathcal{S}(\R^4)$ and set $\vpsi_\lambda\coloneqq \phi(g_\lambda)\Omega$ with $g_\lambda(x)=\lambda^3\chi(\lambda x)$; thus $\hat g_\lambda(k) = \overline{g_\lambda(-k)} = \lambda^{-1} \hat \chi(k/\lambda)$. For the norm $N_\lambda \coloneqq \|\vpsi_\lambda\|$, we find from \eqref{eq:spnorm} after rescaling arguments that
\begin{equation}
   \lim_{\lambda\to\infty} N_\lambda^2 = 
\int\frac{\di^3\textbf{k}}{(2\pi)^3}
\frac{\big\lvert \hat\chi( \|\mathbf{k}\|,\mathbf{k}) \big\rvert^2 }{2\|\mathbf{k}\|}
\end{equation}
by an application of the dominated convergence theorem. (It is used here that $\mathscr{B}^{-}_{\lambda \mathbf{k}} \to 0$, $\mathscr{B}^{+}_{\lambda \mathbf{k}} \to 1$, and $\omega_{\lambda \mathbf{k}}/\|\lambda \mathbf{k}\| \to 1$ in this limit.)

The $L^4$-norm of the states is controlled as follows. We have from Proposition \ref{prop:affilp4norm} that
   \begin{equation}\label{eq:l4upper}
       \|\vpsi_{\lambda}\|_4^4=\|\Delta^{1/4}\phi(g_{\lambda})^\ast\phi(g_{\lambda})\Omega\|
       =  (\phi(g_{\lambda})^\ast\phi(g_{\lambda}) \Omega, \Delta^{1/2} \phi(g_{\lambda})^\ast \phi(g_{\lambda})\Omega)^{1/2},
   \end{equation}
   noting that $\phi(g_{\lambda})^\ast \phi(g_{\lambda})\Omega$ is indeed in the domain of $\Delta^{1/2}$, since $\phi(g_{\lambda})^\ast\phi(g_{\lambda})\eta\M$. Estimating $\Delta^{1/2} \leq 1 + \Delta$, we find
   \begin{equation}
   \begin{aligned}
       \|\vpsi_{\lambda}\|_4^4 &\leq 
       \|\phi(g_{\lambda})^\ast\phi(g_{\lambda})\Omega\|^2 + \|\Delta^{1/2} \phi(g_{\lambda})^\ast \phi( g_{\lambda})\Omega\|^2
       \\&= \|\phi(g_{\lambda})^\ast\phi(g_{\lambda})\Omega\|^2 + \| J\phi(g_{\lambda})^\ast\phi(g_{\lambda})\Omega\|^2
       \\&= 2 \|\phi(g_{\lambda})^\ast\phi(g_{\lambda})\Omega\|^2.
   \end{aligned}
 \end{equation}       
 We can compute the norm on the r.h.s.~explicitly using quasi-freeness of $(\Omega,\cdot \Omega)$; indeed, since crucially $\phi(g_{\lambda})^\ast=\phi(g_{\lambda})$,
 we obtain $ \|\phi(g_{\lambda})^\ast\phi(g_{\lambda})\Omega\|^2 = 3 \|\phi(g_{\lambda})\Omega\|^4$ and thus 
 \begin{equation}\label{eq:psi6est}
       \|\vpsi_{\lambda}\|_4^2 \leq\sqrt{6}N^2_{\lambda},
   \end{equation}
which stays bounded as $\lambda\to\infty$.

Now using \eqref{eq:sph} and rescaling the $x$-, $\mathbf{k}$- and $\mathbf{p}$-integrals, then restricting the $x^0$-integration with positive integrand to an interval $[-R,R]$ before performing the limit $\lambda\to\infty$, we find
\begin{equation}\label{eq:rrescale}
   \liminf_{\lambda\to\infty} (\vpsi_\lambda, h(f) \vpsi_\lambda)
   \geq 2R \,f(0) 
\int\frac{\di^3\textbf{k}}{(2\pi)^3}
\big\lvert \hat \chi( \|\mathbf{k}\|,\mathbf{k}) \big\rvert^2.
\end{equation}
This becomes arbitrarily large when we choose large $R$, thus showing the claim.
\end{proof}

With a similar sequence of states, we can show that the Liouvillian density does not fulfill a lower $L^2$ bound, i.e., no QEI in the usual sense. 

\begin{proposition}\label{prop:nontrivL2}
 Let $f \in \ccs(\R^4)$ be nonnegative and not identically zero. Then the sesquilinear form $\ell(f)$ on $\domP \times \domP$ is not $L^2$-bounded below.
 \end{proposition}
 
  \begin{proof}
 Again we may assume $f(0)>0$ and set $\vpsi_\lambda\coloneqq \phi(g_\lambda)\Omega$, but now with the choice
 \begin{equation}
     \hat g_\lambda (p) = (\bsh{-}{k})^{-1} \lambda^{-1} \hat \chi(k/\lambda)
 \end{equation}
 where $\chi \in\mathcal{S}(\R^4)$ is chosen such that $\hat \chi(k)$ vanishes for $k_0>0$. (A suitable $g_\lambda\in\mathcal{S}(\R^4)$ may be found by Fourier inversion, but it is not real-valued.) Using rescaling as before, one finds from \eqref{eq:spnorm},
 \begin{equation}
     \lim_{\lambda\to\infty} \|\vpsi_\lambda\|^2
     =  
\int\frac{\di^3\textbf{k}}{(2\pi)^3}
\frac{1}{2\| \mathbf{k} \|} \lvert \hat \chi(-\|\mathbf{k}\|,\mathbf{k}) \rvert^2,
 \end{equation}
 thus the norms of the vectors remain bounded. Further, \eqref{eq:sph} yields 
 \begin{equation}
     \lim_{\lambda\to\infty} (\vpsi_\lambda, h(f) \vpsi_\lambda) = 0
 \end{equation}
 due to the decaying $\bsh{-}{k}$-factors. Finally, proceeding from \eqref{eq:spht} as in \eqref{eq:rrescale} yields
 \begin{equation}
   \liminf_{\lambda\to\infty} (\vpsi_\lambda, Jh(f)J \vpsi_\lambda)
   \geq 2R \,f(0) 
\int\frac{\di^3\textbf{k}}{(2\pi)^3}
\big\lvert \hat \chi(-\|\mathbf{k}\|,\mathbf{k}) \big\rvert^2
\end{equation}
for any $R>0$. In summary, the norm of $\vpsi_\lambda$ remains bounded as $\lambda\to\infty$, while $(\vpsi_\lambda,\ell(f)\vpsi_\lambda)\to -\infty$.
\end{proof}

A crucial feature of the above proof is that the functions $g_\lambda$ are \emph{not} real-valued, thus the argument from the proof of Proposition~\ref{prop:nontrivl4}, Eq.~\eqref{eq:psi6est} does not apply, and the present sequence $\vpsi_\lambda$ is not bounded in the $L^4$ norm. This avoids an apparent contradiction with the existence of lower $L^4$-bounds.

In fact, we may put this in a more general context: Remembering that the modular operator can be written as $\Delta=e^{-L}$, we find for suitably regular $A \in \M(\spt)$ that 
\begin{equation}
\begin{aligned}
   (A \Omega, L A \Omega)
   &\geq (A \Omega,A\Omega) - (A \Omega, e^{-L}A\Omega) 
   = (A \Omega,A\Omega) - (\Delta^{1/2}A \Omega, \Delta^{1/2}A\Omega) \\
   &= \|A\Omega\| - \|A^\ast \Omega\|,
\end{aligned}
\end{equation}
where the inequality $x \geq 1-e^{-x}$ has been used. (Cf.~\cite[Sec.~2]{PuszWoronowicz:passive}.) In other words, in order to find vectors with large negative expectation values of $L$, we should not look at $A\Omega$ with selfadjoint $A\in\M(\spt)$, but rather at situations where $A^\ast \Omega$ dominates $A\Omega$ in norm. Our computation above then suggests that the Liouvillian density behaves similarly, noting that
\begin{equation}
\| \phi(g_\lambda)^\ast \Omega \|^2
= \int\frac{\di^3\textbf{k}}{(2\pi)^3}
\frac{1}{2\en{k}}
 \Big(\frac{\bsh{+}{k}}{\bsh{-}{k}}\Big)^2 \lvert \hat g(-k) \rvert^2  \xrightarrow{\lambda\to\infty} \infty.
\end{equation}

\section{The Rindler wedge}\label{sec:rindler}

Our second application example is free field theory on the Rindler wedge, cf.~\cite{Wald:qftcst}. To that end, within Minkowski space $\R^4$, we consider the spacelike wedge region
\begin{equation}
    \rw\coloneqq\left\{x\in\R^4:x^1>|x^0|\right\}.
\end{equation}
Now the spacetime manifold in our framework will not be $\R^4$, but rather $\spt=\rw$ as a globally hyperbolic manifold in its own right. The ``time evolution'' $\Xi_t$ is given by the boost transformations along the wedge, with the Cauchy surface $\csf=\{ (0,\mathbf{x}): x^1>0 \}$. 

Our net of algebras consists again of the usual $C^\ast$-algebras $\A(\ocal)$ of the free real scalar field, but restricted to subregions $\ocal\subset\rw$. On these we consider the usual (Minkowski) vacuum state, which is actually a KMS state for $\beta=2\pi$ with respect to the boost automorphism $\alpha_t$. GNS representation leads us to the usual Fock space $\Hil=\mathcal{F}^s(L^2(\R^3))$, with the vacuum vector $\Omega\in\Hil$ being cyclic and separating for $\M(\rw)=\pi(\A(\rw))''$. The algebras $\M(\ocal)$, $\ocal\subset\rw$, are generated by Weyl operators $\pi(W(f))=\exp i \phi(f)$, $f=\bar f\in\testo$, where $\phi$ is the usual free field,
\begin{equation}
  \phi(x)  = \int \frac{\di^3\mathbf{k}}{(2\pi)^3}\frac{1}{\sqrt{2\en{k}}} \Big( e^{i k\cdot x } \cre{a}{k} +  e^{-i k\cdot x } \ani{a}{k} \Big). 
\end{equation}
It is well known (cf.~\cite{ReedSimon:1975-2}) that $\phi$ is indeed a (boost-)covariant symmetric quantum field associated with $\M$ in the sense of Def.~\ref{def:associated}, with a common and invariant core for all $\phi(f)$ being again the space of finite particle-number Schwartz-class vectors, $\domSchwartz$. 

The modular reflection $J$ implements the wedge reflection $j(x)=(-x^0,-x^1,x^2,x^3)$, that is, we have $J \phi(x) J = \phi(j(x))$ and $\Mt(\ocal) = \M(j^\ast\ocal)$, where the right-hand sides refer to fields and algebras of the full Minkowski space theory (cf.~\cite{Bor:ttt}). The selfadjoint generator of the boost unitaries $U(t)=\Delta^{-it/\pi}$ on $\Hil$ can be written as
\begin{equation}
  K =  i\int \frac{d^3\mathbf{k}}{(2\pi)^3} \Big(  \frac{k_1}{2\en{k}} \cre{a}{k} \ani{a}{k}
+\en{k}(\partial_{1}\cre{a}{k}) \ani{a}{k}
\Big)
\end{equation}
in the sense of sesquilinear forms on $\domSchwartz\times\domSchwartz$. It also has a local density (on Minkowski space), namely 
\begin{equation}\label{eq:kdensity}
    \kappa(x)=\frac{x_1}{2}\left(\sum_{i=0}^3\normord{\partial_i\phi^2}(x)+m^2\normord{\phi^2}(x)\right)-x_0\normord{\partial_0\phi\partial_1\phi}(x),
\end{equation}
in the sense that
\begin{equation}\label{eq:kdensityint}
    (\vpsi_1,K\vpsi_2)=\int_{x^0=0}(\vpsi_1,\kappa(x)\vpsi_2)\di^3\mathbf{x} \quad\text{for all }\vpsi_1,\vpsi_2\in \domSchwartz.
\end{equation}
(This may be checked by a computation similar to Sec.~\ref{sec:integral}.) 

Now in our context, the ``Liouvillian'' of the system is $L=K$, and $\kappa$ is indeed a symmetric quantum field associated with $\M$. (Affiliation with $\M(\ocal')'\cap \M(\rw)$ may be checked as in Sec.~\ref{sec:edaffil}, with $\domSchwartz$ a common core for the $\kappa(f)$, $f\in\testrw$.) However, this refers only to $\kappa(x)$, $x \in \rw$, while the formula \eqref{eq:kdensityint} requires integration of the entire time-0 surface. 

Nevertheless, we can reinterpret the expression in the following way: we set for $x \in \rw$,
\begin{equation}\label{eq:ellwedge}
  \ell(x) \coloneqq \kappa(x) - J\kappa(x)J = \kappa(x) + \kappa(j(x)).    
\end{equation}
For the corresponding ``smeared'' expression $\ell(f)$, it is then clear that $\ell$ is a covariant symmetric quantum field associated with $\Mh$, which is split ($\domSplit=\domSchwartz$) and, due to \eqref{eq:kdensityint}, a density for $L$ in the strong sense.

Much like in the thermal case of Sec.~\ref{sec:freescalar},
but with some important differences in detail, we can deduce $L^4$ inequalities for the ``wedge boost density'' $\ell$. Let $\polyalg(\rw)$ be the polynomial algebra generated by the $\phi(f)$, $f \in \testrw$.

\begin{theorem}
 Let $f = g^2$ with some $g=\bar g \in \testrw \backslash\{0\}$. Then $\ell(f)$ as in \eqref{eq:ellwedge}, as a sesquilinear form on $\polyalg(\rw)\Omega \times \polyalg(\rw)\Omega$, fulfills a nontrivial quantum $L^4(\M(\rw),\Omega)$ inequality.
\end{theorem}
\begin{proof}
  As a first step, we show that $\kappa(f)$ is $L^2$-bounded below (hence also $L^4$-bounded below). Since $\kappa$ is not obviously a ``sum of normal-ordered squares'', this requires some computation, which we defer to \autoref{sec:klower}. On the other hand, $\kappa(f)$ is not $L^4$-bounded above; this is shown in \autoref{sec:knontriv} by exhibiting a suitable sequence of states. Hence $\kappa(f)$ fulfills a nontrivial quantum $L^4$ inequality.

  Now $\polyalg(\rw)\Omega\subset \domAffil{\kappa(j^\ast f)}$ since all operators in $\polyalg(\rw)$ are affiliated with $\M(\rw)$ and have $\polyalg(\rw)\Omega$ as a common invariant dense domain. 
  Thus Theorem~\ref{thm:splitdensity} applies with $\mathcal{D}=\polyalg(\rw)\Omega$, showing that $\ell(f)$ fulfills a nontrivial $L^4$ inequality.
\end{proof}

Also in this case, $\ell(f)$ is $L^4$-bounded below but not $L^2$-bounded below; see Proposition~\ref{prop:nontrivboostL2} in \autoref{sec:knontriv} below.

The Rindler wedge illuminates from a different perspective, and perhaps even more clearly, the origins of our quantum $L^4$ inequality. The ``thermal system'' here consists of all measurements in the wedge $\rw$, while the ``thermal bath'' are the observables in the opposite wedge $\rw'$. System and bath are in a correlated quantum state, the Minkowski vacuum, which is thermal with respect to the boost transformations. The boost density $\kappa$ on $\rw$ takes the role of the energy density on the thermal system, and it fulfills quantum energy inequalities ($L^2$-bounds) in the usual sense. However, in order to construct from $\kappa$ a generator of the boosts with correct action on  \emph{all} observables in the wedge $\rw$, one would need to integrate $\kappa(x)$ right until the edge of the wedge. Due to that sharp cutoff, this procedure (as one may check) does not yield a well-defined selfadjoint operator; rather, one requires the boost density in $\rw'$ as a counterterm, in the sense of Eq.~\eqref{eq:ellwedge}. This necessary counterterm yields extra negative contributions to the Liouvillian density $\ell$, breaking the lower $L^2$ bounds. However, in states generated in the wedge from the ground state (i.e.,~$\psi = A\Omega$, $A\in\M(\rw)$), the contributions from the counterterm are limited -- not with respect to the $L^2$ but to the $L^4$ norm. Thus the resulting Liouvillian density $\ell$ is still $L^4$-bounded below.

\subsection{A lower bound for the boost density}\label{sec:klower}

Our goal in this subsection is to show that the boost generator $\kappa$, as given in \eqref{eq:kdensity}, fulfills a usual quantum inequality (a lower $L^2$-bound) when smeared with $f=g^2\in\testrw$. 

The usual techniques to this end \cite{Few:qft_inequalities} are based on writing the density as sum of (normal-ordered) squares with positive coefficients, which is not the case in Eq.~\eqref{eq:kdensity}. This problem can be circumvented by passing to lightcone coordinates
\begin{equation}
    x'=(u,x,x^2,x^3) = (\frac{x^0+x^1}{2},\frac{x^0-x^1}{2},x^2,x^3),
\end{equation}
noting that $x \in \rw$ if and only if $u>0$ and $v<0$. Using the relations
\begin{equation}
            \frac{\partial}{\partial x^0}=\frac{1}{2}\partial_u+\frac{1}{2}\partial_v,\quad \frac{\partial}{\partial x^1}=\frac{1}{2}\partial_u-\frac{1}{2}\partial_v,
            \quad \operatorname{dVol}_\rw(x) = 2 \di{}u \, \di{}v \, \di{}x^2 \, \di{}x^3, 
\end{equation}
one computes \cite[Proposition 4.3.2]{Sangaletti:thesis} that in the new coordinates,
\begin{equation}\label{eq:kgtransformed}
\begin{aligned}  \kappa(g^2)
       =2 \int g^2(x') \Big(&\frac{u-v}{2}\sum_{i=2}^3\normord{\partial_i\phi^2}(x')+\frac{u-v}{2}m^2\normord{\phi^2}(x')\\
       &+\frac{u}{2}\normord{\partial_v\phi^2}(x')-\frac{v}{2}\normord{\partial_u\phi^2}(x')\;\Big)\di^4 x'.
\end{aligned}
\end{equation}
The integrand is now manifestly written as a sum over normal ordered squares of the field $\phi$ or its derivatives, with coefficients that are squares of smooth real-valued functions, so that we can apply the usual techniques for establishing lower bounds:

\begin{proposition}\label{Prop: state_ind_QEI_for_kr}
    Let $g=\bar g\in\testrw$. Then the sesquilinear form $\kappa(g^2)$ on the domain $\polyalg(\R^4)\Omega\times\polyalg(\R^4)\Omega$ is $L^2$-bounded below.
\end{proposition}

\begin{proof}
We show a lower $L^2$-bound for the term $g^2(x')\frac{u-v}{2}\,\normord{\phi^2}(x')$ within \eqref{eq:kgtransformed}; the others are handled similarly, keeping in mind that $u>0$, $v<0$ in the integrand. Now since $u-v>0$, we know that $h(x'):=\sqrt{(u-v)/2}\, g(x')$ is an element of $\testrw$. Let $\vpsi\in\polyalg(\R^4)\Omega$. Using methods from \cite{FewsterEveson:1998}, we find a $\vpsi$-independent constant $c_h$ such that
\begin{equation}
  \int d^4x' h(x')^2 ( \vpsi, \normord{\phi^2}(x') \vpsi ) \geq - c_h \lVert\vpsi\rVert^2.
\end{equation}
With corresponding bounds on the other terms of the integral \eqref{eq:kgtransformed}, the proposed lower $L^2$-bound is established.
\end{proof}

\subsection{Violation of boost density bounds}\label{sec:knontriv}

We now show violation of upper $L^4$ bounds (nontriviality) and lower $L^2$ bounds for $\ell(f)$. Much like in \autoref{sec:nontriv}, we achieve this by exhibiting a suitable sequence of one-particle states, $\vpsi=\phi(g)\Omega$, $g \in \testrw$. 

\begin{proposition}\label{prop:nontrivboostl4}
      Let $f \in \test$ be nonnegative and not identically zero. Then the sesquilinear form $\kappa(f)$ on the domain $\polyalg(\rw)\Omega \times \polyalg(\rw)\Omega$ is not $L^4$-bounded above.
\end{proposition}

\begin{proof}
Using boost covariance, as well as translation covariance in $x^2$ and $x ^3$, we may assume that $f(0,a,0,0)>0$ for some $a>0$. Fix $\chi=\bar \chi\in\testrw$ and set $\vpsi_\lambda\coloneqq \phi(g_\lambda)\Omega$ with $g_\lambda(x)=\lambda^3\chi(\lambda x - \lambda ae_1)$, or equivalently, $\hat g_\lambda(k) = \lambda^{-1} \hat\chi(\lambda^{-1} k) e^{-ik_1a}$. One finds that
\begin{equation}
   \lim_{\lambda\to\infty} \|\vpsi_\lambda\|^2 = 
\int\frac{\di^3\textbf{k}}{(2\pi)^3}
\frac{\big\lvert \hat\chi(\|\mathbf{k}\|,\mathbf{k}) \big\rvert^2}{2\|\mathbf{k}\|};
\end{equation}
the $L^4$-norm of these vectors remains bounded as well, since
$\|\vpsi_{\lambda}\|_4^2\leq  \sqrt{6} \|\vpsi_{\lambda}\|_2^2$ with the same argument leading to \eqref{eq:psi6est}, which applies since $\chi$ is real-valued.

We aim to show that $( \vpsi_\lambda, \kappa(f) \vpsi_\lambda) \to \infty$ as $\lambda\to\infty$. To that end, one first notes that on inserting \eqref{eq:kgtransformed}, the expectation value $ (\vpsi_\lambda, \kappa(f) \vpsi_\lambda )$ is a sum of five positive terms, each arising from one of the normal-ordered squares, and it suffices to show that one of them diverges to $+\infty$; we pick the one that involves $\normord{(\partial_2 \phi)^2}$. That is, we consider (after converting back to the coordinates $x$) the expression
\begin{equation}
  K_2(f) \coloneqq
  \int \di{}^4x \,f(x) \int \di^3\mathbf{k} \,\di^3\mathbf{p} \frac{x_1 k_2p_2}{2(2\pi)^6 \en{k}\en{p}} \overline{\hat g_\lambda(k)} \hat g_\lambda(p) e^{i(k-p)\cdot x} 
\end{equation}
and show that it diverges. Indeed, after rescaling the integration variables, we have
\begin{equation}
  K_2(f) = 
  \int \di{}^4x \,\big(\frac{x_1}{\lambda}+a\big) f(\lambda^{-1}x + ae_1) \int \frac{ k_2p_2}{2 (2\pi)^6 \omega_{\lambda \mathbf{k}}\omega_{\lambda \mathbf{p}}/\lambda^2} \overline{\hat \chi(k)} \hat \chi(p) e^{i(k-p)\cdot x} .
\end{equation}
Noticing that the $x$-integrand is positive, we find for any $R>0$,
\begin{equation}
  \liminf_{\lambda \to\infty} K_2(f) \geq \int_{-R}^R \di{x}_0 \, a f(ae_1) \int \frac{\di^3\mathbf{k}}{(2\pi)^3}\frac{k_2^2}{\|\mathbf{k}\|}  \big\lvert\hat\chi( \|\mathbf{k}\|,\mathbf{k} )\big\rvert^2
\end{equation}
which becomes arbitrarily large for large $R$, showing the claim.
\end{proof}

This directly leads us to nonexistence of lower $L^2$ bounds:
\begin{proposition}\label{prop:nontrivboostL2}
 Let $g=\bar g \in \ccs(\R^4)$ be not identically zero. Then the sesquilinear form $\ell(g^2)$ on $\polyalg(\rw)\Omega \times \polyalg(\rw)\Omega$ is not $L^2$-bounded below.
 \end{proposition}
 \begin{proof}
   We know from Proposition~\ref{Prop: state_ind_QEI_for_kr} that $\kappa(g^2)$ is $L^2$-bounded below on the stated domain. Hence the statement is equivalent to $J\kappa(g^2) J$ \emph{not} being $L^2$-bounded above. 

  Assume, for sake of contradiction, that $J\kappa(g^2) J$ is $L^2$-bounded above on $\polyalg(\rw)\Omega \times \polyalg(\rw)\Omega$. Applying Proposition~\ref{Prop: state_ind_QEI_for_kr}, and noting that $\polyalg(\R^4)\Omega$ is invariant under $J$, we know that $J\kappa(g^2) J$ is also $L^2$-bounded below on  $\polyalg(\rw)\Omega \times \polyalg(\rw)\Omega$, hence bounded there.
  Now $J\polyalg(\rw)\Omega$ is dense in $\Hil$ by the Reeh-Schlieder property of the real scalar free field. For the operator $\kappa(g^2)\restriction J\polyalg(\rw)\Omega$, this means that its adjoint, and hence its double adjoint (closure), are both defined on all of $\Hil$ and bounded. In particular, the original operator $\kappa(g^2)$ defined on $\domSchwartz$ is bounded. But this is a contradiction to Proposition~\ref{prop:nontrivboostl4}.     
\end{proof}

\section{Conclusions}\label{sec:conclusions}

We have seen that in quantum field theory, the Liouvillian density, or the density of the (negative of the) modular generator, can possess an ``almost positivity'' property akin to a quantum energy inequality (QEI). Naturally this inequality does not apply to all states uniformly, but specifically to those generated from the base state by applying an observable of the thermal system (rather than of the heat bath). This is reminiscent of the passivity of the Liouvillian, or negativity of the modular generator on corresponding subspaces. Mathematically, the property can be expressed as lower bounds with respect to the noncommutative $L^4$ norm.

We have discussed the fundamental approach to these features in a model-independent setting. However, as with QEIs, the question of existence of lower bounds can be answered only in specific models. To that end, we have considered the free real scalar field on Minkowski space in its usual thermodynamical equilibrium state, as well as the vacuum sector of the same model restricted to the Rindler wedge, where the boost density takes the role of the density of the Liouvillian. This choice was mainly for simplicity: We expect that similar results hold in more general linear theories, also on other static spacetimes and in general space-time dimensions. Challenges still need to be overcome, though, for a generalization to models with self-interaction: Our proof techniques are partially based on the existence of state-independent QEIs, which are not yet well understood beyond the linear case; see however \cite{FewsterHollands:2005,BostelmannFewster:2009,BostelmannCadamuroFewster:ising,BostelmannCadamuroMandrysch:qeiint}.

As the Rindler wedge example shows, the concept of $L^p$ inequalities applies not only to thermal equilibrium in the usual sense, but more generally to the modular generator in Tomita-Takesaki theory. Since the importance of the latter in relativistic quantum information has recently attracted attention (e.g.,  \cite{HollandsSanders:entanglement,CasiniGrilloPontello:relent,CeyhanFaulkner:qnec,CLR:waveinfo}), we expect our results to be relevant in that context as well.

\appendix

\section{Operators affiliated with a von Neumann algebra}\label{app:affil}

While unbounded (closed) operators cannot be \emph{contained} in a von Neumann algebra, there is still a closely related notion of \emph{affiliation}. We recall this notion, and show some auxiliary results that are not directly available in the literature.

Let $\M$ be a von Neumann algebra acting on a Hilbert space $\Hil$. 
A densely defined, closed operator $T$ is called \emph{affiliated with $\M$} (in symbols, $T \affil \M$) if it fulfills one (hence both) of the following equivalent conditions:
\begin{enumerate}[(i)]
 \item \label{it:affil1} For every $A \in \M'$, we have $A \dom T \subset \dom T$ and $AT\vpsi=TA\vpsi$ for all $\vpsi \in \dom T$;
 \item \label{it:affil3} If $T = V \int_0^\infty \lambda dE(\lambda)$ is the polar decomposition of $T$, then $V, E(\lambda)\in \M$. 
\end{enumerate}
We briefly recall the arguments leading to this equivalence:
\begin{enumerate}
 \item[(i)$\Rightarrow$(ii):] If $U\in\M'$ is unitary, hence $UTU^\ast=T$, then uniqueness of the polar and spectral decompositions yield $UVU^\ast=V$, $UE(\lambda)U^\ast=E(\lambda)$. Every element of $\M'$ can be written as a linear combination of unitaries, thus $V,E(\lambda)\in\M''=\M$. \cite[Lemma 2.5.8.]{BraRob:qsm1} 
 \item[(ii)$\Rightarrow$(i):] Consider the sequence $T_n := T E(n)$ in $\M$. We note that $E(n)\vpsi' \oplus T_n \vpsi'$ is in the graph of $T$ for every $\vpsi'\in\Hil$. Thus for any $A \in \M'$ and $\vpsi \in \dom{T}$, we have 
 \begin{equation}
     E(n) A\vpsi \oplus T_n A \vpsi = E(n) A \vpsi \oplus A T_n  \vpsi \xrightarrow{n \to \infty} A \vpsi \oplus A T \vpsi    
 \end{equation}
 within the (closed) graph of $T$, due to continuity of $A$. Thus $A\vpsi \in \dom{T}$ and $TA\vpsi=AT\vpsi$.
\end{enumerate}

A weaker form of condition (\ref{it:affil1}) is already sufficient to show affiliation (cf.~also \cite[Remark~5.6.3]{KadRin:algebras1}):
\begin{lemma}\label{lem: weaker_affiliation}
   Let $\mathcal{D}_0$ be a core for the closed operator $T$, and let $\N_0\subset\M'$ be a weakly dense $\ast$-subalgebra.
   Suppose that for every $\vpsi \in \mathcal{D}_0$ and $A \in \N_0$, we have $A\vpsi \in \dom T$ and $AT\vpsi=TA\vpsi$. Then $T \affil \M$.
\end{lemma}
\begin{proof}
  We may assume that $\N_0$ is unital, since $A=\idop$ fulfills the condition above. 
  
  Now for given $\vpsi \in \dom{T}$, choose $(\vpsi_n)\subset \mathcal{D}_0$ with $\vpsi_n \to \vpsi$ and $T \vpsi_n \to T\vpsi$. Likewise for given $A \in \M'$, choose a net $(A_\lambda)\subset \N_0$ such that $A_\lambda \xrightarrow[]{\lambda} A$ strongly (note that $\N_0\subset \M'$ is also strongly dense). Then the following convergence statements hold within the (closed) graph of $T$, using continuity of $A_\lambda$:
  \begin{equation}
      A_\lambda \vpsi_n \oplus TA_\lambda \vpsi_n 
      =  A_\lambda \vpsi_n \oplus A_\lambda T \vpsi_n
      \xrightarrow{n \to \infty} A_\lambda \vpsi \oplus A_\lambda T \vpsi
      \xrightarrow{\lambda} A \vpsi \oplus A T \vpsi.
  \end{equation}
  We conclude that $A\vpsi \in \dom{T}$ and $TA\vpsi = AT\vpsi$. Thus condition (\ref{it:affil1}) is fulfilled.
\end{proof}

Vice versa, a stronger form of (i) holds that involves two closed (possibly unbounded) operators:
\begin{lemma}\label{lem: affaff}
    Suppose that $T \affil \M$ and $T' \affil \M'$. If $\vpsi\in \dom{T}\cap \dom{T'}$ and $T\vpsi\in\dom{T'}$, then $T'\vpsi \in \dom{T}$ and
    \begin{equation*}
        TT'\vpsi=T'T\vpsi.
    \end{equation*}
\end{lemma}
\begin{proof}
   Let $T_n' := T' E'(n)$, where $V',E'(\lambda)\in\M'$ arise from polar decomposition of $T'$ as in (\ref{it:affil3}). Then $T_n'\vpsi \to T'\vpsi$ since $\vpsi\in\dom{T'}$ and 
   \begin{equation}
       T T_n'\vpsi = T_n'T\vpsi \to T'T \vpsi
   \end{equation}
   since $T\vpsi \in \dom{T'}$. By closedness of $T$ (notice that $T_n'\vpsi\in\dom{T}$ since $T\affil \M$ and $T'_n\in\M '$), we now have $T'\vpsi\in\dom{T}$ and $TT'\vpsi = T'T\vpsi$.
\end{proof}

If $\M$ has a cyclic and separating vector $\Omega\in\Hil$, and hence the Tomita operator $S$ is defined by
    \begin{equation}
        SA\Omega=A^*\Omega, \quad \A \in \M
    \end{equation}
and operator closure, we can extend this relation to affiliated operators:
\begin{proposition}\label{prop:Ext_Tomita}
    If $T\affil\mathcal{M}$ and $\Omega\in\dom{T} \cap \dom{T^*}$, then we have $T\Omega\in\dom{S}$ and
    \begin{equation*}
        ST\Omega=T^*\Omega.
    \end{equation*}
\end{proposition}

\begin{proof}
  Let $V$, $E(\lambda)$ as in (\ref{it:affil3}), and set $T_n:=TE(n)\in\M$. It is known that $T^\ast = V^\ast T V^\ast$,
  where $\Omega\in\dom{T^\ast}$ implies $V^\ast\Omega \in \dom{T}$ \cite[Sec.~7.1]{Schmuedgen:unbounded}; similarly $T_n^\ast = V^\ast T_n V^\ast$. Then, within the graph of $S$,
  \begin{equation}
    T_n \Omega \oplus S T_n \Omega =
    T_n \Omega \oplus V^\ast T_n V^\ast \Omega
    \xrightarrow{n \to \infty} T \Omega \oplus V^\ast T V^\ast \Omega
    = T \Omega \oplus T^\ast \Omega,
  \end{equation}
  thus $T\Omega \in \dom{S}$ and $ST\Omega=T^\ast \Omega$.
\end{proof}

We now discuss the sum of two closed operators, say, $T_1$ and $T_2$, which we take to be naturally defined on $\dom{T_1}\cap \dom{T_2}$. In general, this sum $T_1+T_2$ need not be densely defined nor closable. However, if it is both, then the sum operation works as expected with respect to affiliation:
\begin{lemma}\label{lem:sumaffil}
Let $T_1$, $T_2$ be two closed operators affiliated with von Neumann algebras $\M_1$ and $\M_2$ respectively, and suppose that $T_1+T_2$ is densely defined (on $\dom T_1 \cap \dom T_2$) and closable. Then $\overline{T_1+T_2}$ is affiliated with $\M_1 \vee \M_2$.
\end{lemma}
\begin{proof}
Let $A\in (\M_1 \vee \M_2)'=\M_1' \cap \M_2'$. Then, by characterization (i) of affiliation, we have $A \dom{T_j}\in \dom{T_j}$ and $AT_j=T_jA$ on $\dom{T_j}$ ($j=1,2$). 
Hence for $\vpsi\in\mathcal{D}_\cap := \dom{T_1} \cap \dom{T_2}$, we have $A\mathcal{D}_\cap\subset \mathcal{D}_\cap$ and
\begin{equation}
   \overline{T_1+T_2} A \vpsi = T_1 A \vpsi + T_2 A \vpsi = A T_1 \vpsi +AT_2 \vpsi = A (\overline{T_1+T_2})\vpsi. 
\end{equation}
Since $\mathcal{D}_\cap$ is by definition a core for $\overline{T_1+T_2}$, the result now follows from Lemma~\ref{lem: weaker_affiliation}.
\end{proof}

\section{Noncommutative \texorpdfstring{$L^p$}{Lp} spaces}\label{app:lp}

Let $\M$ be a von Neumann algebra of operators acting on Hilbert space $\mathcal{H}$, with a cyclic and separating normalised vector $\Omega$.
The noncommutative $L^p$ spaces $L^p(\M,\Omega)$, $2 \leq p \leq \infty$, are Banach spaces that interpolate between $L^2(\M,\Omega)\coloneqq \Hil$ (with the norm of $\Hil$)  and $L^\infty(\M,\Omega)\coloneqq \M$ (with the operator norm), continuously included in each other, i.e., $L^p(\M,\Omega) \subset L^q(\M,\Omega)$ for $p \geq q$. Their name is justified since they share many properties with the classical $L^p$ spaces, notably the Hölder inequality. 

These spaces have first been constructed for tracial von Neumann algebras \cite{Dixmier:lp}, and are still best known in that setting. In our applications, however, $\M$ is a factor of type III, so that we require a more generic formulation. Several approaches to this end have been proposed in the literature, all of which turn out to be equivalent: by extending $\M$ to a larger algebra that possesses a trace \cite{Haagerup}; by complex interpolation between $L^2(\M,\Omega)$ and $L^\infty(\M,\Omega)$ \cite{KOSAKIci}; and using the relative modular operators associated with $\M$ \cite{ArakiMasuda}. If $\M=C(X)$ is commutative, and $\Omega$ a constant function on $X$, one recovers the classical spaces $L^p(X)$.

For the convenience of the reader, let us give here a direct construction of the space $L^4(\M,\Omega)$, which is particularly important in our context. We start with the space $\Hil = L^2(\M,\Omega)$, the norm of which we denote by $\|\cdot\|_2$. The Tomita operator of $\M,\Omega$ on $\Hil$ will be denoted $S = J \Delta^{1/2}$. By cyclicity, the subspace $\M\Omega \subset \Hil$ is dense; on it, we define a new norm by
\begin{equation} \label{eq:noncnormdef}    
\begin{aligned}
    \|\cdot\|_4:\;\;\;&\M\Omega\to \R,\\
     &A\Omega \mapsto \|A\Omega\|_4 \coloneqq \|\Delta^{1/4}A^*A\Omega\|_2^{1/2}.
\end{aligned}
\end{equation}
This is well-defined since $A^\ast A \Omega \in \dom{S}$. That it fulfills the axioms of a norm is less direct: 
\begin{lemma}
    The real-valued function $\|\cdot\|_4$ indeed defines a norm.
\end{lemma}
\begin{proof}
Absolute homogeneity is an immediate consequence of the absolute homogeneity of the norm $\|\cdot\|_2$. The proof of subadditivity requires some computation. Using subadditivity of $\|\cdot\|_2$, we obtain for $A,B\in\M$,
\begin{equation}\label{eq: proof_triang}
     \begin{aligned}
     \|A\Omega+B\Omega\|_4^2 &= 
     \|\Delta^{1/4}(A^*+B^*)(A+B)\Omega\|_2
    \\
    &\leq \|\Delta^{1/4}A^*A\Omega\|_2+\|\Delta^{1/4}B^*B\Omega\|_2+\|\Delta^{1/4}(A^*B+B^*A)\Omega\|_2.
     \end{aligned}
 \end{equation}
 In addition, since $A$ commutes with $JAJ$, we have
 \begin{equation}\label{eq:JAJ}
     \| AJA\Omega \|^2
     = (AJA \Omega, AJAJ\Omega)
     = (JA^\ast J JA \Omega, A^\ast A \Omega)
     = (\Delta^{1/2} A^\ast A \Omega,A^\ast A \Omega)
     = \| A \Omega\|_4^4, 
     \end{equation}
     and similarly for $B$, from which we conclude
 \begin{equation}
 \begin{aligned}
\left(\Delta^{1/4}A^*B\Omega,\Delta^{1/4}A^*B\Omega\right)&=\left|\left(JBJB\Omega,AJA\Omega\right)\right|,
      \\
      &\leq 
      \| A \Omega\|_4^2 \| B \Omega\|_4^2
      \\ \left(\Delta^{1/4}B^*A\Omega,\Delta^{1/4}B^*A\Omega\right)& \leq 
      \| B \Omega\|_4^2 \| A \Omega\|_4^2,
 \\
 \operatorname{Re}{\left(\Delta^{1/4}A^*B\Omega,\Delta^{1/4}B^*A\Omega\right)}&\leq \left|\left(\Delta^{1/4}A^*B\Omega,\Delta^{1/4}B^*A\Omega\right)\right|
      \\
     &\leq  \|\Delta^{1/4}A^*B\Omega\|_2\|\Delta^{1/4}B^*A\Omega\|_2
     \\&\leq  
      \| A \Omega\|_4^2 \| B \Omega\|_4^2.
 \end{aligned}
 \end{equation}
 Therefore
 \begin{equation}
     \|\Delta^{1/4}(A^*B+B^*A)\Omega\|_2\leq 2\| A \Omega\|_4 \| B \Omega\|_4.
 \end{equation}
 Substituting into \eqref{eq: proof_triang}, we obtain
 \begin{equation}
     \|A\Omega+B\Omega\|_4^2\leq \|A\Omega\|_4^2+\|B\Omega\|_4^2+2\|A\Omega\|_4\|B\Omega\|_4,
 \end{equation}
 which concludes the proof of subadditivity.
 
 Finally, using positive-definiteness of the norm $\|\cdot\|_2$, we observe that $\|A\Omega\|_4=0$ implies
 \begin{equation}
     0 =\|\Delta^{1/4}A^*A\Omega\|^{1/2}_2\implies \Delta^{1/4}A^*A\Omega=0.
 \end{equation}
 Since $A^*A\Omega\in\dom{S}$, we can deduce that
 \begin{equation}
     0=\Delta^{1/2}A^*A\Omega= SA^*A\Omega=A^*A\Omega,
 \end{equation}
and since $\Omega$ is separating, that implies $A^\ast A=0$, hence $A=0$. This shows that $\|\cdot\|_4$ is positive definite and, in conclusion, a norm on the vector space $\M\Omega$. 
 \end{proof}
 
 Now the normed vector space $(\M\Omega,\|\cdot\|_4)$ can be embedded into the Hilbert space $(\Hil,\|\cdot\|_2)$ via the inclusion map $\iota$:
 \begin{equation}
 \begin{aligned}
     \iota:\;\;\;(\M\Omega,\|\cdot\|_4)&\to(\Hil,\|\cdot\|_2)\\
      A\Omega&\mapsto A\Omega.
 \end{aligned}
 \end{equation}
 This inclusion map is continuous since
 \begin{equation}
 \begin{aligned}
     \|\iota (A\Omega)\|_2&=(A\Omega,A\Omega)^{1/2}=(A^*A\Omega,\Omega)^{1/2}=(\Delta^{1/4}A^*A\Omega,\Omega)^{1/2}\\
     &\leq \|\Delta^{1/4}A^*A\Omega\|_2^{1/2}\cdot 1=\|A\Omega\|_4.
\end{aligned}
 \end{equation}
 We can now complete the normed vector space $(\M\Omega,\|\cdot\|_4)$ to a Banach space, which we denote $L^4(\M,\Omega)$.
 
 Since the inclusion map $\iota$ is continuous, we can extend it by continuity to the entire Banach space $L^4(\M,\Omega)$. More precisely, the extension $\hat{\iota}$ is defined as follows: If $\vpsi \in L^4(\M,\Omega)$ is the $\|\cdot\|_4$-limit of $\vpsi_n \in \M\Omega$, then $\hat \iota(\vpsi)=\operatorname{\Hil-lim} \vpsi_n$.  
One can show that the extended map $\hat\iota$ is still injective (see \cite[Prop.~3.2.6]{Sangaletti:thesis} for the proof) and, therefore, that $L^4(\M,\Omega)$ is continuously embedded into  $L^2(\M,\Omega)=\Hil$.

 The above construction of the Banach space $L^4(\M,\Omega)$ suffices for all we need in this article, hence the reader may confine themselves to it. For completeness, we show that is equivalent to the noncommutative $L^4$ space known in the literature, specifically, in the formulation of Araki and Masuda \cite{ArakiMasuda}. To that end, it suffices to show that the noncommutative $L^4$ norm of \cite{ArakiMasuda} (denoted in the following as $\|\cdot\|_4^{\mathrm{AM}}$) coincides with our definition \eqref{eq:noncnormdef} on the dense subspace $\M \Omega$. According to \cite{ArakiMasuda} it is given by
\begin{equation}\label{eq:l4am}
    \|A\Omega\|_4^{\mathrm{AM}}=\sup_{\|\vxi\|=1}\|\Delta^{1/4}_{\vxi,\Omega}A\Omega\|_2,
\end{equation}
where $\Delta_{\vxi,\Omega}$ is the relative modular operator as defined in Appendix C of the same paper (where the properties used in the following are also proven). Since the relative modular operator depends on the vector $\vxi$ only through the state $\omega_\vxi(\cdot)\coloneqq\left(\vxi,\cdot\vxi\right)$ it defines on $\M$, and since every normal state has exactly one vector representative in the positive cone $V_\Omega^{1/4}$ (of unit norm), we can restrict the supremum in \eqref{eq:l4am} to $\vxi\in V_\Omega^{1/4}$. Since also $\Omega$ is an element of the cone, we have $J=J_{\Omega,\Omega}=J_{ \vxi,\Omega}$ and therefore:
\begin{equation}
     \|A\Omega\|_4^{\mathrm{AM}}=\sup_{\vxi \in V_\Omega^{1/4}, \|\vxi\|=1}\left(A^*\vxi,JA\Omega\right)^{1/2}=\sup_{\vxi \in V_\Omega^{1/4}, \|\vxi\|=1}\left(\vxi,AJA\Omega\right)^{1/2}.
\end{equation}
Note that the vector $AJA\Omega$ belongs to the positive cone $V_\Omega^{1/4}\subset\mathcal{H}$ as well. Therefore, using again \eqref{eq:JAJ},
\begin{equation}    \|A\Omega\|_4^{\mathrm{AM}}=\frac{(AJA\Omega,AJA\Omega)^{1/2}}{\|AJA\Omega\|_2^{1/2}}=\|AJA\Omega\|_2^{1/2} = \| A\Omega \|_4
\end{equation}
(cf.~\cite{Hollands:trace_ineq}, Eq.~(151) with $n=4$, $\psi=\eta=\Omega$).

Finally we show how, under certain domain conditions, vectors obtained by applying operators \emph{affiliated} with $\M$ to the vacuum vector $\Omega$ are contained in the noncommutative $L^4$ space as well, with the formula \eqref{eq:noncnormdef} for their noncommutative $L^4$ norm remaining valid.

\begin{proposition}\label{prop:affilp4norm}
  Let $A\affil\mathcal{M}$ and $\Omega\in\dom{A}$. If $A\Omega\in\dom{A^*}$, then $A\Omega \in L^4(\M,\Omega)$ and
    \begin{equation*}
        \|A\Omega\|_4^2=\|\Delta^{1/4}A^*A\Omega\|_2.
    \end{equation*}
\end{proposition}
\begin{proof}
Let $A=V\int_{{\R}^+}\lambda \di{}E(\lambda)$ be the polar decomposition. Consider the sequence $A_n\coloneqq AE(n)$ in $\M$. We start showing that:
\begin{equation}
    A_n\Omega\xrightarrow[L^4]{n\to\infty}\vpsi
\end{equation}
for some $\vpsi\in L^4$. Since $L^4$ is a Banach space, it is enough to show that the sequence $A_n\Omega$ is Cauchy in the $L^4$ norm. 
%Specifically, we show that given $\epsilon>0$, there exists an $N\in\mathbb{N}$ such that, for every $n,m\geq N$, 
%\begin{equation*}
%    \|(A_n-A_m)\Omega\|_4^2=\|\Delta^{1/4}(A_n-A_m)^*(A_n-A_m)\Omega\|_2<\epsilon^2.
%\end{equation*}
To that end, notice that for $m<n$,
\begin{equation}\label{eq:anamdiff}
\begin{aligned}
        \|(A_n-A_m)\Omega\|^2_4 &= \|\Delta^{1/4}(A_n-A_m)^*(A_n-A_m)\Omega\|_2\\
        &= \big(\Delta^{1/4}(A_n-A_m)^*(A_n-A_m)\Omega,\Delta^{1/4}(A_n-A_m)^*(A_n-A_m)\Omega\big)^{1/2}\\
        &\leq \big((A_n-A_m)^*(A_n-A_m)\Omega,(\Delta+\idop)(A_n-A_m)^*(A_n-A_m)\Omega\big)^{1/2}\\
        &=\sqrt{2}\|(A_n-A_m)^*(A_n-A_m)\Omega\|_2,
    \end{aligned}
    \end{equation}
    where we used the estimate $\Delta^r\leq\Delta+\idop$ valid for $0\leq r\leq 1$, also as sesquilinear forms. By spectral calculus we obtain:
    \begin{equation}
  \|(A_n-A_m)^*(A_n-A_m)\Omega\|_2 \leq \left(\int_0^\infty\lambda^4\chi_{[m,\infty)}\left(\Omega,\di E(\lambda)\Omega\right)\right)^{1/2}.
  \end{equation}
   The condition $\Omega\in\dom (A^*A)$ allows us to apply the dominated convergence theorem. Since the characteristic functions $\chi_{[m,\infty)}$ converges pointwise to $0$ in the limit $m\to\infty$, we obtain the Cauchy property of the sequence $A_m \Omega$ in $L^4$.
   
   We now compute the $L^4$ norm of the limit vector $\vpsi$. By continuity of the norm we have:
\begin{equation}\label{eq:normcont}
    \|\vpsi\|^2_4=\lim_{n\to\infty}\|A_n\Omega\|^2_4=\lim_{n\to\infty}\|\Delta^{1/4}A_n^*A_n\Omega\|_2.
\end{equation}
We want to show that the sequence $\Delta^{1/4}A_n^*A_n\Omega$ converges in $\Hil$ to the vector $\Delta^{1/4}A^*A\Omega$. (Note that the latter is well-defined thanks to Proposition~\ref{prop:Ext_Tomita}, with $T=A^\ast A$; the domain assumptions in our hypothesis enter.) With a computation as in \eqref{eq:anamdiff} we find
\begin{equation}
    \|\Delta^{1/4}(A_n^*A_n-A^*A)\Omega\|_2\leq \sqrt{2}\|(A^*_nA_n-A^*A)\Omega\|_2.
\end{equation}
Using the spectral theorem it follows that
\begin{equation}
    \|(A^*_nA_n-A^*A)\Omega\|_2=\int_0^{\infty}\lambda^4\theta (\lambda-n)\left(\Omega,\di E(\lambda)\Omega\right).
\end{equation}
The condition $\Omega\in\dom{(A^*A)}$ guarantees once more that the dominated convergence theorem applies and, therefore, one finds convergence (in $\Hil$) of the sequence $\Delta^{1/4}A_n^*A_n\Omega$ to the vector $\Delta^{1/4}A^*A\Omega$. With \eqref{eq:normcont}, we conclude
\begin{equation}
    \|\vpsi\|_4^2=\|\Delta^{1/4}A^*A\Omega\|_2.
\end{equation}

As a final step, we show $\vpsi=A\Omega$ or, more precisely, $\hat{\iota}(\vpsi)=A\Omega$. By definition of the extended inclusion $\hat{\iota}$ we have:
\begin{equation}
    \hat{\iota}(\vpsi)=\operatorname*{\Hil-lim}_{n\to\infty} \,A_n\Omega=A\Omega,
\end{equation}
where the last equality can once more be proven using the spectral theorem together with the condition $\Omega\in\dom{A}$.  
\end{proof}

\section*{Acknowledgements}
D.C.\ and L.S.\ have been supported by the Deutsche Forschungsgemeinschaft (DFG) within the Emmy Noether grant CA1850/1-1. 

\bibliographystyle{alpha}
\bibliography{qei}

\end{document}